\documentclass{article}
\usepackage{graphicx,amsthm, amsmath,amssymb,amsfonts, dsfont, mathtools} 
\usepackage{tikz}
\usepackage{cleveref}
\usepackage{biblatex} 
\usepackage{thmtools}
\usepackage{thm-restate}
\usepackage{booktabs}
\usepackage[a4paper,margin=2.5cm]{geometry}
\usepackage{algorithm}
\usepackage{algpseudocode}

\usetikzlibrary{positioning,arrows.meta,calc} 

\DeclareMathOperator*{\argmin}{arg\,min}

\newtheorem{definition}{Definition}[section]

\newtheorem{remark}{Remark}[section]

\newtheorem{lemma}{Lemma}[section]
\newtheorem{assumption}{Assumption}
\newtheorem{corollary}{Corollary}

\title{Partial Identification Learning with Categorical Treatments for Individualized Treatment Rules}
\author{Johannes Hruza\thanks{University of Copenhagen. Corresponding author: \texttt{johannes.hruza@sund.ku.dk}}, Paweł Morzywołek\footnotemark[1], Jakob Zeitler \thanks{University of Oxford}, Samir Bhatt\footnotemark[1], Michael Sachs\footnotemark[1]}
\date{March 2025}

\begin{document}

\maketitle

\begin{abstract}
We develop a partial identification learning framework for individualized treatment rules (ITRs) with categorical treatments, outcomes, and instrumental variables. Rather than relying on strong causal assumptions required for point identification, our framework leverages causal bounds to characterize the optimal treatment decision. Existing methods for ITR optimization under partial identification are largely restricted to binary treatment settings and the bounds derived by Balke and Pearl under the canonical instrumental variable design. We extend this framework to accommodate a broader class of causal structures as well as scenarios with categorical treatment, outcome, and instrumental variables. We introduce a generalized minimax loss criterion for treatment selection from among more than two options, which minimizes the maximum possible difference between the chosen and the optimal treatment based on partial identification bounds. To construct the ITR, we use a symmetric embedding strategy that maps discrete treatments to the vertices of a regular simplex, avoiding the geometric inconsistencies of standard one-vs-rest approaches. We derive a differentiable, weighted surrogate risk function and show that optimizing it solves the original problem. Furthermore, we provide finite sample convergence rates via an oracle inequality under general regularity conditions, which we show are satisfied by a kernel based implementation. Numerical experiments demonstrate that the framework yields ITRs significantly closer to the oracle ITR compared to existing alternatives in settings with unmeasured confounding.
\end{abstract}


\section{Introduction}

An individualized treatment regime (ITR) aims to construct a decision rule that assigns the most suitable option from a set of possible actions based on an individual’s observed characteristics. The study of optimal ITRs has gained considerable attention across statistics, biomedicine, and econometrics. While the terminology varies (treatment regime, policy learning, treatment choice) the problem is universally formulated as optimizing expected outcomes (or welfare) under heterogeneous treatment effects \cite{Murphy2003,Manski2004,Kitagawa2018,Athey2021,Robins2004,Zhang2014,Tsiatis2019}. 

A critical assumption in much of the early works on ITRs is no unmeasured confounding, i.e. the assumption that all variables influencing both treatment assignment and the outcome are measured. Unfortunately, this assumption is rarely guaranteed in observational studies and it is frequently violated in randomized clinical trials due to noncompliance. When unmeasured confounding is present, standard ITR estimators can be biased, leading to suboptimal or even harmful policy recommendations.

One approach to circumvent the need to assume no unmeasured confounding is the use of instrumental variables (IVs). An IV is a variable that influences the treatment assignment but is independent of unmeasured confounders and has no direct effect on the outcome. Classic examples include randomized assignment in trials with imperfect compliance \cite{Angrist1996} or genetic variants in epidemiology, where the random assortment of alleles serves as a natural instrument (refered to as as Mendelian randomization) \cite{Davies2018, Lawlor2008}.

The information provided by an IV can be leveraged in two distinct ways. First, if one is willing to make strong structural assumptions -- such as monotonicity of the treatment effect or specific functional forms -- the IV can be used to point identify causal estimands like the Local Average Treatment Effect (LATE) \cite{Angrist1996, Imbens1994}. However, these assumptions are often untestable and may be implausible in complex medical or economic settings.

Alternatively, one can adopt a \textit{partial identification} approach. Instead of enforcing implausible assumptions necessary for point estimation, this method uses IVs to bound the potential outcomes, defining a ``feasible region'' within which the true causal effect must be \cite{Manski1990, Balke1997}. Recent advances have extended these bounding techniques to general causal graphs and automated derivations \cite{Sachs2022, Duarte2023, Zhang2021}. In the context of ITRs, these bounds can be used to define a minimax criterion: finding a policy that minimizes the bound-based risk consistent with the observed data, thereby ensuring robustness against unmeasured confounding of arbitrary strength.

In previous work, Pu and Zhang \cite{puzhang2021} successfully integrated this partial identification idea into ITR learning. They developed a framework to estimate optimal policies using a binary instrumental variable, reformulating the task as a weighted binary classification problem solved via Support Vector Machines (SVM). However, their approach -- like most of the ITR literature -- is limited to binary treatment settings. This limitation is significant because many real-world decision problems are inherently multicategorical (e.g., choosing among multiple competing drugs or dosage levels).

In this paper, we introduce a \textit{Categorical Partial Identification Learning} framework that extends the existing binary theory to accommodate scenarios with multiple (finite) treatment options and categorical instrumental variables. Unlike prior methods that rely on simple weighted classification, we propose a risk minimization approach using a simplex embedding. This allows us to map the categorical decision problem into a vector space, enabling the learning of smooth, nonlinear decision boundaries that are robust to unmeasured confounding.

The remainder of the paper is outlined as follows: In \cref{sec:problem-setup}, we formalize the problem setup and define the partial identification bounds for categorical treatments. \Cref{sec:partial-identification-learning} introduces our learning objective using partial identification. This is followed by \cref{sec:symetric-embedding}, where we propose our learning algorithm and the use of the simplex embedding strategy. In \cref{sec:theoretical_properties}, we provide theoretical guarantees for the learner. \Cref{sec:experiments} presents numerical experiments demonstrating the method's robustness and behavior in various settings. This is followed by a discussion in \cref{sec:discussion}.

\section{Problem Setup}
\label{sec:problem-setup}

\subsection{Notation}

Let the observed data consist of $N$ independent and identically distributed samples. We assume that there is an observed set of variables $\mathcal{O}$ that contains a vector of pre-treatment covariates $\boldsymbol{X} \in \mathcal{X} \subseteq \mathbb{R}^d$, a discrete action (or treatment) $A \in \mathcal{A} = \{1, \ldots, k\}$, and a categorical outcome $Y$. For a realization of $\boldsymbol{X}$, we write $\boldsymbol{x} \in \mathcal{X}$. Without loss of generality, we assume throughout this paper that a higher value of $Y$ is more favorable.

We frame the problem within the Nonparametric Structural Equation Models (NPSEM) framework \cite{Pearl2009}. Let $\Gamma$ be a Directed Acyclic Graph (DAG) with a vertex set $\mathcal{V}$ containing both the observed variables $\mathcal{O}$ and unobserved variables. While the specific structure of $\Gamma$ may vary, we assume at a minimum that $\boldsymbol{X}$ is a parent of both $A$ and $Y$, and that $A$ is a parent of $Y$. Under this framework, the value of each variable $V \in \mathcal{V}$ is determined by a deterministic structural function $h_{V}$:
\begin{align}
    v = h_{V}(\text{PA}_{V}, u_{V}),
\end{align}
where $\text{PA}_{V}$ denotes the realized values of the parents of $V$ in the graph $\Gamma$, and $u_{V}$ represents exogenous noise variables that are jointly independent. We make no assumptions regarding the functional form of $h_{V}$ or the distribution of unmeasured confounding variables (which may be multivariate or continuous).

To define causal effects, we use potential outcome notation derived directly from these structural equations. Let $Y(a)$ denote the potential outcome of $Y$ that would be observed if the action $A$ were set to value $a$. Formally, this corresponds to the value of $Y$ in a modified structural model where the equation for $A$ is replaced by the constant $A=a$, while all other equations remain unchanged:
\begin{align}
    Y(a) \coloneqq h_Y(\text{PA}_Y\setminus \{A\}, a , u_Y).
\end{align}
To guide individualized decision-making, we aim to determine the conditional expected outcome under a specific intervention $a$ for a patient with covariates $\boldsymbol{x}$, defined as

\begin{align}
    \mu_a(\boldsymbol{x}) := \mathbb{E}[Y(a) \mid \boldsymbol{X}=\boldsymbol{x}].
\end{align}

\subsection{Partial Identification}\label{sec:partial-identification}

Identifying the conditional expected potential outcome $\mu_a(\boldsymbol{x})$ typically requires the assumption of no unmeasured confounding. In observational studies, where treatment assignment may be influenced by unobserved factors $U$, this assumption is often violated, making $\mu_a(\boldsymbol{x})$ unidentifiable from the observed data distribution $P(\mathcal{O})$ alone.

Consequently, we adopt a \textit{partial identification} framework. Rather than relying on strong, often untestable assumptions to force point identification, we strictly characterize the feasible region of the parameter of interest compatible with the observed data and the causal model. Formally, for each action $a \in \mathcal{A}$ and covariate value $\boldsymbol{x} \in \mathcal{X}$, we identify (ideally sharp) lower and upper bounds, $L_a(\boldsymbol{x})$ and $U_a(\boldsymbol{x})$, such that:
\begin{align}
L_a(\boldsymbol{x}) \leq \mu_a(\boldsymbol{x}) \leq U_a(\boldsymbol{x}). \label{eq:ptwsvalid}
\end{align}
It is important to note that the width of this interval, $[L_a(\boldsymbol{x}), U_a(\boldsymbol{x})]$, represents the identification uncertainty inherent to the problem structure. Unlike a statistical confidence interval, this width does not vanish even when given the full distribution $P(\mathcal{O})$.

Our learning framework is agnostic to the specific method used to derive these bounds, it requires only that the bounds are pointwise valid under the assumed causal structure, i.e., equation \eqref{eq:ptwsvalid} holds for all $\boldsymbol{x}$ and $a$ \cite{Jonzon2025}. However, the usefulness of the ITR we want to derive from these bounds depends heavily on the width of the intervals. Instrumental Variables (IVs) are frequently employed in this context as they can significantly narrow the bounds compared to assumption-free scenarios. 

The literature provides various techniques for deriving such bounds depending on the available information of the data. For instance if we have, after conditioning on $\boldsymbol{X}$, a valid DAG structure (which is not guaranteed for example if $\boldsymbol{X}$ is a collider), Manski and Pepper \cite{Manski2000} derive bounds under Monotone Instrumental Variable assumptions; Levis et al. \cite{Levis2025} formulate bounds in a potential outcome framework; and Balke and Pearl \cite{Balke1997} established seminal sharp symbolic bounds for binary outcomes with a binary instrument. We stress that by assuming the validity of the identification bounds, we implicitly adopt the assumptions required by the specific method used to derive them.

Even though our method works with any valid bounding method, in the following, we use the package \verb|causaloptim| by Sachs et al. \cite{Sachs2022} to allow for bounding of $\mu_a(\boldsymbol{x})$ for categorical treatments and IVs under a broad class of causal structures. This approach translates the constraints of the NPSEM into a linear optimization problem to solve symbolically for the narrowest possible bounds given the graph structure and constraints. While we employ this specific algorithm for our experiments, the framework presented in this paper remains valid for any set of identification intervals derived from a valid causal model.

\section{Partial Identification Learning for Categorical Treatment}
\label{sec:partial-identification-learning}
\subsection{Risk Optimality}
An ITR, or policy, is formally a measurable map $g: \mathcal{X} \to \mathcal{A}$ that assigns a treatment action to each subject based on their observed pre-treatment covariates. Let $\mathcal{D}$ denote the class of all such measurable functions.  The expected potential outcome under ITR $g$ is defined as $\mathbb{E}[Y(g(\boldsymbol{X}))] = \mathbb{E}[\mu_{g(\boldsymbol{X})}(\boldsymbol{X})]$.

Our objective is to find a policy that maximizes this value. Ideally, if the true conditional means $\mu_a(\boldsymbol{x})$ were known, the optimal policy would simply be:
\begin{align}
\label{eq:optimal_regime}
    g^*_O(\boldsymbol{x}) := \arg\max_{a \in \mathcal{A}} \mu_a(\boldsymbol{x}), \quad \forall \boldsymbol{x} \in \mathcal{X}.
\end{align}

Which would also minimize the oracle risk of a policy defined below.

\begin{definition}
The \textbf{oracle loss} of choosing action $a$ given covariates $\boldsymbol{x}$ is the difference between the conditional means of the best possible outcome (the oracle outcome) and the potential outcome under $a$:
\begin{align}
\label{eq:oracle-loss}
\mathcal{L}^*(\boldsymbol{x},a) = \max_{b \in \mathcal{A}} \mu_b(\boldsymbol{x}) - \mu_{a}(\boldsymbol{x}).
\end{align}
Let $\mathcal{D}$ be the set of measurable functions from $\mathcal{X}$ to $\mathcal{A}$. The \textbf{oracle risk} of a policy $g \in \mathcal{D}$ is:
\begin{align}
\label{eq:oracle-risk}
\mathcal{R}^*(g) := \mathbb{E}_{\boldsymbol{X}}[\mathcal{L}^*(\boldsymbol{X}, g(\boldsymbol{X}))],
\end{align}
\end{definition}

Since $\mathcal{R}^*$ relies on the unknown conditional mean of the potential outcome $\mu_a(\boldsymbol{X})$, we cannot optimize it directly. However, we can construct a conservative upper bound for it. As summarized in the previous section, partial identification allows us to bound $\mu_a(\boldsymbol{X})$ without the strong assumptions required for full identification, with its true value guaranteed to be  within $[L_a(\boldsymbol{X}), U_a(\boldsymbol{X})]$.

This range of possible values for the conditional potential outcome mean complicates the definition of loss and hence the optimization of the ITR. For binary treatments, Cui \cite{Cui2021} summarizes different decision-making criteria for selecting treatment under uncertainty due to unmeasured confounding. The concept of optimality under partial identification is related to a particular choice of decision-making criterion. 

First, the Optimist (Maximax) criterion maximizes the best case scenario ($\max_a U_a(\boldsymbol{X})$). This strategy is suitable for situations where high uncertainty is acceptable to achieve a potentially high reward. Second, the Pessimist (Maximin) criterion maximizes the worst-case lower bound ($\max_a L_a(\boldsymbol{X})$). This conservative approach prioritizes safety and the ``do no harm" principle, ensuring a guaranteed minimum outcome, but often at the cost of ignoring treatments that have high potential but wider bounds. Finally, the Opportunist (Minimax) criterion seeks to minimize the maximum possible loss of making a wrong decision. Rather than focusing solely on the best or worst absolute outcomes, it focuses on the loss relative to the optimal choice. 

In this work, we adopt the Opportunist perspective. We aim to protect against the worst-case error we might make relative to the unknown optimal treatment.

\begin{definition}
For $\boldsymbol{x} \in \mathcal{X}$ and an action $a \in \mathcal{A}$, the \textbf{worst-case loss} is:
\begin{align}
\label{eq:worst-case-loss}
\mathcal{L}(\boldsymbol{x}, a) := \max_{b \in \mathcal{A}\setminus \{a\}} U_b(\boldsymbol{x}) - L_a(\boldsymbol{x}).
\end{align}
Let $\mathcal{D}$ be the set of measurable functions from $\mathcal{X}$ to $\mathcal{A}$. The \textbf{worst-case risk} of a policy $g \in \mathcal{D}$ is:
\begin{align}
\label{eq:worst-case-risk}
\mathcal{R}(g) := \mathbb{E}_{\boldsymbol{X}}[\mathcal{L}(\boldsymbol{X}, g(\boldsymbol{X}))],
\end{align}
\end{definition}
Intuitively, $\mathcal{L}(\boldsymbol{x}, a)$ quantifies \textit{``If I choose treatment $a$, what is the most I could have lost by not choosing the best alternative?''}. It compares the lower bound of our choice $L_a$ against the upper bound of the best alternative $U_b$. If $\mathcal{L}(\boldsymbol{x}, a) \leq 0$, then action $a$ is better than all alternatives because its lower bound exceeds their upper bounds. This approach balances the risk of loss against the potential for gain, making it a robust choice in settings with high uncertainty where we wish to avoid errors without being too conservative. This idea is also illustrated in \cref{fig:worst-case-loss}.

Our objective is to find a policy $ g^* \in \mathcal{D} $ that minimizes the risk  \eqref{eq:worst-case-risk}.
\begin{definition}
Consider the infimum
\begin{align*}    
\inf_{g \in \mathcal{D}} \mathcal{R}(g).
\end{align*}
Any policy $ g \in \mathcal{D} $ that attains this infimum is called a minimax (or bound optimal) policy.
\end{definition}

\begin{restatable}{prop}{minimadecision}
\label{prop:minimax-decision}
Define the map:
\begin{align*}
    g^* \colon \mathcal{X} &\to \mathcal{A}\\
    \boldsymbol{x} &\mapsto \argmin_{a \in \mathcal{A}} \mathcal{L}(\boldsymbol{x}, a),
\end{align*}

where ties are broken according to a fixed ordering of the finite action set $\mathcal{A}$ (i.e., $g^*$ selects the smallest minimizer with respect to this ordering).
Then $g^*$ is a minimax rule and satisfies $\mathcal{R}(g^*) = \inf_{g \in \mathcal{D}} \mathcal{R}(g)$.
\end{restatable}

As $g^*$ is the most straightforward candidate that satisfies the criterion of a minimax policy, it will serve as a benchmark for simulations and the theoretical analysis for learned methods.

The difference between minimizing the (worst-case) risk and the oracle risk is determined by the width of the partial identification bounds.

\begin{restatable}{prop}{riskerror}
\label{prop:risk-error}
Let $W(a,\boldsymbol{X}) := U_a(\boldsymbol{X})-L_a(\boldsymbol{X})$ be the width of the bounds. For any $g\in \mathcal{D}$:
    \begin{align*}
    & \quad \mathcal{R}(g)-\mathcal{R}^*(g)\\
    &\leq \mathbb{E} \left[\max_{b \in \mathcal{A} \setminus \{g(\boldsymbol{X})  \}} W(b,\boldsymbol{X})+W(g(\boldsymbol{X}),\boldsymbol{X})\right]\\
    &\leq  \mathbb{E}\left[ \sum_{a \in \mathcal{A}} W(a,\boldsymbol{X})\right].
    \end{align*}
\end{restatable}
This shows that narrower bounds (for example, obtained by stronger IVs) directly shrink the gap between the estimable and the theoretically optimal policy.

\begin{remark}
\label{rmk: CATE-connection}
In the binary case ($k=2$), our formulation reduces to the frameworks stated in Cui \cite{Cui2021} and Pu \& Zhang \cite{puzhang2021}, which is only defined for a binary action space. Specifically, optimizing \cref{eq:worst-case-risk} is equivalent to optimizing the conditional average treatment effect (CATE) based ``risk function'' $\mathcal{R_{\text{upper}}}$ (Definition 1 of Pu \& Zhang \cite{puzhang2021}). See \cref{apx:proofs} for further elaboration. 
\end{remark}

\begin{figure}[ht]
    \centering
    \begin{tikzpicture}[scale=0.7, every node/.style={font=\small}]
        \draw[->] (0,0)-- (14.5,0); 
        
        \foreach \x in {2, 3, 6, 7, 8, 10, 11, 13} {
            \draw (\x,0.1) -- (\x,-0.1) node[below] {\x};
        }
    
        \foreach \L/\R/\Y/\P/\Lab in {
              2/8/1/3/4,
              2/7/2/6/3,
              6/10/3/6/2,
              11/13/4/13/1
        }{
            \draw (\L,\Y) -- (\R,\Y);
            \draw[thick] (\L,\Y-0.2) -- (\L,\Y+0.2);
            \draw[thick] (\R,\Y-0.2) -- (\R,\Y+0.2);
            \fill (\P,\Y) circle (0.1);
            \node[right] at (\R+0.2,\Y) {Treatment $\Lab$};
        }
    \end{tikzpicture}
    
    \caption{Visualization of the worst-case loss. Circles denote the (unknown) true means $\mu_a$, whereas bars denote bounds $[L_a, U_a]$. Treatment 1 is clearly optimal ($L_1 > U_{\max\{2,3,4\}}$), its worst-case loss is negative as well as the oracle loss. 
    If we disregard treatment 1, the bounds of the remaining treatments $\{2,3,4\}$ overlap. In this case, treatment $4$ has a worst-case loss of $\max(U_2, U_3) - L_4 = 8$ (and oracle loss of $3$). Selecting 2 has a worst-case loss of $\max(U_3, U_4) - L_2 = 8-6=2$ (and oracle loss of $0$).} 
    \label{fig:worst-case-loss}
\end{figure}
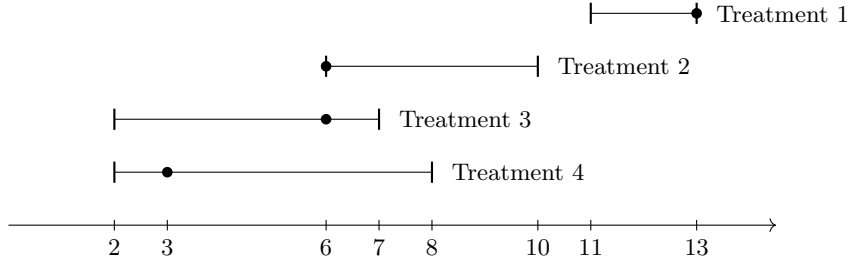

While the minimax policy $g^*$ minimizes the worst case risk, directly optimizing the empirical counterpart of $\mathcal{R}(g)$ is computationally difficult due to discontinuities caused by jumps of the worst-case loss, leading to undefined or vanishing gradients. This prevents the use of efficient optimization methods using gradient descent.

In their binary treatment setting $(k=2)$, Pu et al.  \cite{puzhang2021} address this by first reformulating the empirical minimax problem as a weighted classification problem with an indicator loss, and then replace the loss containing indicator functions by a surrogate loss (hinge-loss) which reformulating the problem as a weighted Support Vector Machine (SVM).
However, common multiclass extensions of SVMs, such as solving a series of binary problems by solving one class vs all the others (One vs Rest), are not suited for this setting. First, One vs Rest decomposes the problem into $k$ independent binary tasks. This does not reflect our objective: minimizing $k$ independent binary losses does not imply minimizing the worst-case loss $\mathcal{L}(\boldsymbol{X}, a)$, which relies on the relative comparison between the chosen action and the best alternative. Second, geometrically, One vs Rest approaches often result in ambiguous decision regions where multiple classifiers predict ``positive" or all predict ``negative" \cite{Lee2004}. Finally, these methods typically require learning $k$ functions to model a decision boundary with only $k-1$ degrees of freedom, resulting in a redundancy (the probabilities of each class have to sum to one) and a lack of the symmetric geometric interpretation.

\section{Symmetric Multicategory Embedding}
\label{sec:symetric-embedding}

To overcome the limitations of disjoint binary classifiers, we adopt the multicategory angle-based large-margin classification framework \cite{Zhang2014}. The core motivation is to move from a set of independent decision boundaries to a single, unified latent space representation. 

Since categorical treatments have no inherent ordinal structure, the decision space should treat all $k$ actions symmetrically. We achieve this by embedding the actions into the vector space $\mathbb{R}^{k-1}$ as the vertices of a regular simplex. This representation is minimal in the sense that it is using exactly $k-1$ dimensions to represent the degrees of freedom among $k$ classes and ensures that the decision boundaries are geometrically consistent, avoiding the redundancies common in One vs Rest approaches \cite{Lee2004}.

Following \cite{Zhang2014, Zhang2020}, each treatment level $a \in \mathcal{A}$ is enumerated from $1$ to $k$ and represented by a fixed vector $\boldsymbol{w}_a:=\phi(a)\in \mathbb{R}^{k-1}$ through the map:

\begin{align}
\label{eq: immersion}
\phi \colon \{1, \dots, k\} &\to \mathbb{R}^{k-1}\\
a &\mapsto \begin{cases}
(k-1)^{-1/2}\zeta, & a=1, \\
-\frac{1+k^{1/2}}{(k-1)^{3/2}}\zeta + \left(\frac{k}{k-1}\right)^{1/2}e_{a-1}, & 2 \leq a \leq k, \notag
\end{cases}    
\end{align}
where $\zeta = (1,\dots,1)$ is a vector of length $k-1$ with all elements set to 1, while $e_j$ is the $j$-th standard basis vector in $\mathbb{R}^{k-1}$ (i.e., its $j$-th element is 1 and all others elements are 0).

The resulting set of action vectors $\{\boldsymbol{w}_1, \dots, \boldsymbol{w}_k\}$, are the vertices of a regular simplex with $k$ vertices in a $(k-1)$-dimensional space. This simplex is centered at the origin, and each vertex vector $\boldsymbol{w}_j$ has a Euclidean norm of one and the angle between any two action vectors is equal \cite{Zhang2014}. 

Using this simplex embedding, the learning objective is to find a mapping $f$ that maps the vector of covariates $\boldsymbol{x}$ into this latent space in a way that reflects treatment optimality in terms of the angular alignment. We look for a function $f$ such that for a subject with covariates $\boldsymbol{x}$, the mapping $f(\boldsymbol{x})$ points in the direction of the vertex $\boldsymbol{w}_a$ corresponding to the best treatment. This formulation relaxes the hard, discrete problem of treatment selection into a smooth optimization problem, where the optimal policy is induced by maximizing the angular alignment between the embedded covariate vectors and the vectors representing the different levels of the treatment.

The decision rule for a given embedding $f$ is then given by:
\begin{align}
\label{eq:policy}
g_f(\boldsymbol{x}) := \arg\max_{a \in \mathcal{A}} \langle \boldsymbol{w}_a, f(\boldsymbol{x}) \rangle, \mbox{ for all } \boldsymbol{x} \in \mathcal{X},
\end{align}
where $\langle \cdot, \cdot \rangle$ is the euclidean inner product.

By maximizing the inner product, we essentially assign $\boldsymbol{x}$ to the action vector with the smallest angle to $f(\boldsymbol{x})$ in the latent space. This avoids the redundancy of sum to zero constraints found in other multiclass classifiers while maintaining a geometric interpretation. The setup is illustrated in \cref{fig: illustration}.
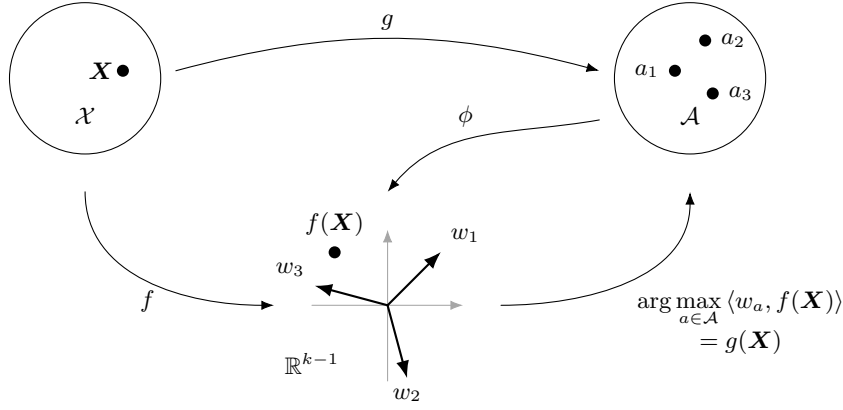
\begin{figure}
    \centering

    \begin{tikzpicture}[
        >=Latex,
        line cap=round,
        line join=round,
        every node/.style={font=\small},
        dot/.style={circle, fill=black, inner sep=1.6pt}
    ]

    \draw (-4, 3) circle (1);
    \node at (-4,2.5) {$\mathcal{X}$};
    \node[dot] (xpt) at (-3.5,3.1) {};
    \node[left] at (-3.5,3.1) {$\boldsymbol{X}$};

    \draw[->] (-2.8,3.1) to[out=15,in=165] node[midway, above] {$g$} (2.8,3.1);
    \draw[->] (-4,1.5) to[out=270,in=180] node[midway, below] {$f$} (-1.5, 0);
    \draw[->] (1.5, 0) to[out=0,in=270]node[midway, below right, align=center] {$\arg\max\limits_{a\in\mathcal{A}}\,\langle w_a,f(\boldsymbol{X})\rangle$ \\ $= g(\boldsymbol{X})$} (4, 1.5);
    \draw[->] (2.8,2.45) to[out=190,in=50] node[midway, above left] {$\phi$} (0,1.5);
    
    \draw (4,3) circle (1);
    \node at (4, 2.5) {$\mathcal{A}$};
    
    \node[dot,label=left:$a_1$] at (3.8,3.1) {};
    \node[dot,label=right:$a_2$] at (4.2,3.5) {};
    \node[dot,label=right:$a_3$] at (4.3,2.8) {};

    \coordinate (o) at (0,0);
    
    \draw[->,gray!70] (-1,0) -- (1,0);
    \draw[->,gray!70] (0, -1) -- (0,1);
    
    \draw[thick,->] (o) -- ++(0.707, 0.707) node[above right] {$w_1$};
    \draw[thick,->] (o) -- ++(0.259, -0.966) node[below] {$w_2$};
    \draw[thick,->] (o) -- ++(-0.966, 0.259) node[above left] {$w_3$};
    
    \node[dot,label=above:$f(\boldsymbol{X})$] at (-0.7,0.7) {};
    \node[below] at (-1,-0.5) {$\mathbb R^{k-1}$};

    \end{tikzpicture}

    \caption{Illustration of the simplex embedding for a treatment with $k=3$ levels. The treatment levels $a_1,a_2,a_3\in\mathcal{A}$ are mapped by $\phi$ to the vectors $w_1,w_2,w_3\in\mathbb R^{k-1}=\mathbb R^2$, which form the vertices of a regular simplex. A vector of covariates $\boldsymbol{X}\in\mathcal{X}$ is mapped by $f$ into the same latent space, and the induced policy assigns the treatment level with maximal alignment, $g(\boldsymbol{X})=\arg\max_{a \in \mathcal{A}}\langle w_a,f(\boldsymbol{X}) \rangle$.}    
    \label{fig: illustration}
\end{figure}

\subsection{Weighted Surrogate Loss}

Our objective is to find a function $f$ such that the induced policy $g_f$ minimizes the worst-case risk $\mathcal{R}(g)$. However, minimizing the empirical risk directly is infeasible due to the discontinuity of the $\max$ operator in \eqref{eq:worst-case-loss}.

To address this, we construct a smooth, differentiable surrogate loss by considering a weighted aggregation of the worst-case loss. We use the softmax function to determine the relative weight $p_a(\boldsymbol{x}, f)$ assigned to each action's worst-case loss based on the embedding $f(\boldsymbol{x})$:
\begin{align}
\label{eq:softmax-weights}
p_a(\boldsymbol{x}, f) = \frac{\exp(\langle \boldsymbol{w}_a, f(\boldsymbol{x}) \rangle)}{\sum_{b=1}^k \exp(\langle \boldsymbol{w}_b, f(\boldsymbol{x}) \rangle)}.
\end{align}
These weights $p_a$ dictate how much the worst-case loss of action $a$ contributes to the surrogate risk. An embedding $f(\boldsymbol{x})$ that is well-aligned with $\boldsymbol{w}_a$ results in a large weight $p_a \approx 1$, meaning the surrogate loss at that point is dominated by the worst-case loss of action $a$. Conversely, misaligned actions receive negligible weights.

We define the surrogate loss as the expected worst-case loss under this weighting. For the estimation we substitute the worst-case loss $\mathcal{L}(\boldsymbol{x},a)$ with the empirical worst-case loss $\hat{\mathcal{L}}(\boldsymbol{x},a)$. 

\begin{definition}
Let $\boldsymbol{\mathcal{L}}(\boldsymbol{x)} := (\mathcal{L}(\boldsymbol{x}, 1), \dots, \mathcal{L}(\boldsymbol{x}, k))^T$ and $\boldsymbol{p}(\boldsymbol{x},f):= (p_1(\boldsymbol{x}, f) , \dots , p_k(\boldsymbol{x}, f))$. For a function $f$ we define:
\begin{align}
\label{eq:loss-function}
\mathcal{L}^s(\boldsymbol{x}, f) &\coloneqq \boldsymbol{p}(\boldsymbol{x},f)\cdot \boldsymbol{\mathcal{L}(x)}  &\text{(Surrogate Loss)}
\\
\mathcal{R}^s(f) &\coloneqq \mathbb{E}[\mathcal{L}^s(\boldsymbol{X}, f)] & \text{(Surrogate Risk)}\\
\label{eq:emp-loss-function}
\hat{\mathcal{R}}^s(f) &:= \frac{1}{N} \sum_{i=1}^N \left[ \sum_{a=1}^k p_a(\boldsymbol{X}_i, f) \hat{\mathcal{L}}(\boldsymbol{X}_i, a) \right] + \lambda \Omega(f) & \text{(Empirical Regularized Surrogate Risk)},
\end{align}
where $\Omega(f)$ is a regularization term and $\lambda$ a tuning parameter.
\end{definition}
This surrogate risk forces the model to assign higher weights (and thus larger inner products) to actions with lower worst-case losses. We note that the surrogate risk is $m$ times differentiable with respect to the parametrization of $f$,  whenever the regularization term $\Omega(f)$ is $m$ times differentiable for some $m\in \mathbb{N} \cup\ \{\infty\}$. However the surrogate loss is not convex in general. We note that the surrogate loss in \cref{eq:loss-function} is a special case of the batch policy optimization objective considered by \textcite{Chen2019}, where the action specific reward is given by the negative worst-case loss. In their setting, the nonconvexity is addressed by adding entropy regularization to the expected risk and optimizing a related strongly convex objective. A similar construction may also be useful in our setting.


We summarize the proposed estimation procedure in Algorithm \ref{alg:cpil}. 
The procedure is agnostic both to the method of bounding the partial identification intervals and to the choice of learner used to estimate the embedding function $f$. In particular, $f$ may be estimated using a kernel method, a neural network, or any other suitable method.

\subsection{Kernelized Formulation}
\label{sec:kernelized-formulation}

We model the embedding function $f$ within a Reproducing Kernel Hilbert Space (RKHS), since optimal treatments often depend on complex interactions between patient covariates. In addition to its ability to be flexible enough to handle these nonlinear interaction, the RKHS norm enforces smoothness, effectively regularizing the policy against the sharp, conservative ``spikes" induced by partial identification bounds. Crucially, this framework allows us to apply the Representer Theorem, reducing the optimization from a infinite dimensional function space into a finite dimensional problem.

Let $\mathcal{H}$ be a RKHS of real-valued functions with a positive definite kernel $k(\cdot, \cdot)$ and induced inner product $\langle \cdot, \cdot\rangle_{\mathcal{H}}$. Since our target embedding $f: \mathcal{X} \to \mathbb{R}^{k-1}$ is vector-valued, we construct the product space $\mathcal{H}^{k-1} = \mathcal{H} \times \cdots \times \mathcal{H}$. We equip this space with the inner product defined as the sum of the component wise inner products:
\begin{align*}
\langle f, g \rangle_{\mathcal{H}^{k-1}} := \sum_{j=1}^{k-1} \langle f^j, g^j\rangle_\mathcal{H}, \quad \text{for } f,g \in \mathcal{H}^{k-1}.
\end{align*}
The induced norm is given by $\| f \|_{\mathcal{H}^{k-1}} = \sqrt{\langle f, f \rangle_{\mathcal{H}^{k-1}}}$.

We decompose the embedding $f$ into $k-1$ components, $f(\boldsymbol{X})=(f^1(\boldsymbol{X}),\dots, f^{k-1}(\boldsymbol{X}))$, and restrict each component $f^j$ to lie in $\mathcal{H}$. We define the regularization term $\Omega(f)$ as the squared norm in the product space:
\begin{align}
\Omega(f) := \| f \|^2_{\mathcal{H}^{k-1}} = \sum_{j=1}^{k-1} \| f^j \|^2_\mathcal{H}.
\end{align}

Since the regularization term is additive across components, we can apply the Representer Theorem (see Schölkopf et al. \cite{Scholkopf2001}, Theorem 4.2) to each component $f^j$ individually. Specifically, fixing the other components does not influence the optimization of $f^j$. Consequently, the minimizer for each component $j = 1, \dots ,k-1$ admits the finite expansion:
\begin{align}
f^j(X) = \sum_{i=1}^N \alpha_{ij} k(X_i, X),
\end{align}
where $\alpha_{ij} \in \mathbb{R}$ are learnable coefficients.

This allows us to rewrite the optimization problem in terms of a coefficient matrix. Let $\boldsymbol{A} \in \mathbb{R}^{N \times (k-1)}$ be the matrix of coefficients of $\alpha_{ij}$ and denote $\boldsymbol{\alpha}_j \in \mathbb{R}^N$ the $j$-th column of $\boldsymbol{A}$. Furthermore, let $\mathbb{K} \in \mathbb{R}^{N \times N}$ be the kernel Gram matrix with entries $\mathbb{K}_{lm} = k(\boldsymbol{X}_l, \boldsymbol{X}_m)$, and let $\boldsymbol{k}_i \in \mathbb{R}^N$ denote the $i$-th column of $\mathbb{K}$.

Using the expansion $f^j(\boldsymbol{X}_l) = \sum_{i=1}^N \alpha_{ij} k(\boldsymbol{X}_i, \boldsymbol{X}_l) = \boldsymbol{\alpha}_j^T \boldsymbol{k}_l$, we can derive the kernelized form of the regularization term using the reproducing property:
\begin{align}
\Omega(f) &= \sum_{j=1}^{k-1} \| f^j \|^2_\mathcal{H} \nonumber 
= \sum_{j=1}^{k-1} \sum_{i=1}^N \sum_{l=1}^N \alpha_{ij} \alpha_{lj} \langle k(\boldsymbol{X}_i, \cdot), k(\boldsymbol{X}_l, \cdot) \rangle_\mathcal{H} \nonumber \\
&= \sum_{j=1}^{k-1} \boldsymbol{\alpha}_j^T \mathbb{K} \boldsymbol{\alpha}_j \nonumber 
= \operatorname{Tr}(\boldsymbol{A}^T \mathbb{K} \boldsymbol{A}).
\end{align}

Substituting these expressions into our objective, we arrive at the final finite-dimensional optimization problem. Let $\hat{\mathcal{L}}_{ia} = \hat{\mathcal{L}}(\boldsymbol{X}_i, a)$ be the precomputed worst-case loss. The kernelized surrogate risk $\hat{\mathcal{R}}^s$ to minimize with respect to $\boldsymbol{A}$ is:
\begin{align}
\label{eq:ker_risk}
\min_{\boldsymbol{A} \in \mathbb{R}^{N \times (k-1)}} \quad \frac{1}{N} \sum_{i=1}^N \left( \frac{\sum_{a=1}^k \hat{\mathcal{L}}_{ia} \exp\left( \boldsymbol{w}_a^T \boldsymbol{A}^T \boldsymbol{k}_i \right)}{\sum_{b=1}^k \exp\left( \boldsymbol{w}_b^T \boldsymbol{A}^T \boldsymbol{k}_i \right)} \right) + \lambda \operatorname{Tr}(\boldsymbol{A}^T \mathbb{K} \boldsymbol{A}).
\end{align}

Since the risk in this setting is smooth we can employ gradient based optimization such as the L-BFGS algorithm \parencite{Liu1989} which has been used in the simulations  in \cref{sec:experiments}. 

\subsection{Neural Network Parameterization}
\label{sec:neural-network}

While the kernelized formulation in \cref{sec:kernelized-formulation} provides theoretical guarantees via the finite representation due to the Representer Theorem, it scales poorly to massive datasets. Evaluating the kernelized policy requires computing the kernel function between any new data point and every sample in the training set, leading to inference times that scale linearly with the training size $N$. Furthermore, constructing and storing the $N \times N$ Gram matrix $\mathbb{K}$ requires $\mathcal{O}(N^2)$ memory.

To scale our categorical partial identification framework, we alternatively model the embedding function $f: \mathcal{X} \to \mathbb{R}^{k-1}$ using a deep neural network. Let $f_\theta(\boldsymbol{X})$ denote a Multilayer Perceptron (MLP) \parencite{Bengio2016} parameterized by weights $\theta$. Rather than projecting the data into an implicit reproducing kernel Hilbert space, the neural network learns an explicit, nonlinear mapping into the $(k-1)$-dimensional simplex space. Because the network's parameters are fixed after training, inference time is entirely independent of $N$. 

We optimize the parameters $\theta$ by directly minimizing the empirical surrogate loss $\hat{\mathcal{L}}(f_\theta)$ (\cref{eq:emp-loss-function}) using stochastic gradient descent.

\begin{algorithm}
\caption{Pseudo-algorithm for Partial Identification Learning}
\label{alg:cpil}
\begin{algorithmic}[1]
\Statex \textbf{Input:} Observed data containing covariates $\boldsymbol{X}_i$, action $A_i$, and outcome $Y_i$, together with any additional observed variables (such as IVs) and assumptions required to construct valid partial identification bounds; action set $\mathcal{A}=\{1,\dots,k\}$; simplex embedding vectors $\{\boldsymbol{w}_a\}_{a=1}^k$; regularization parameter $\lambda$.

\Statex \textbf{Output:} Estimated policy $\hat g:\mathcal{X}\to\mathcal{A}$.

\Statex

\State Split the data into two disjoint sets $I=I_1 \sqcup I_2$. Using a valid bounding procedure and train models that estimate bounds for $\hat L_a(\boldsymbol{x})$ and $\hat U_a(\boldsymbol{x})$ using $I_1$ for each action $a\in\mathcal{A}$ and $\boldsymbol{x} \in \mathcal{X}$.

\State For each datapoint in $I_2$ and each action, compute the empirical worst-case loss
\begin{align*}
\hat{\mathcal{L}}(\boldsymbol{X}_i,a) := \max_{b\in\mathcal{A}\setminus\{a\}} \hat U_b(\boldsymbol{X}_i)-\hat L_a(\boldsymbol{X}_i).
\end{align*}

\State Using $I_2$ learn $\hat{f}_s:\mathcal{X}\to\mathbb{R}^{k-1}$ by minimizing the empirical surrogate risk
\begin{align*}
\hat{\mathcal{R}}^s(f)\coloneqq \frac{1}{|I_2|}\sum_{i\in I_2}\sum_{a=1}^kp_a(\boldsymbol{X}_i,f)\,\hat{\mathcal{L}}(\boldsymbol{X}_i,a)+\lambda \Omega(f),
\end{align*}
where
\begin{align*}
p_a(\boldsymbol{X},f)=\frac{\exp(\langle \boldsymbol{w}_a,f(\boldsymbol{X})\rangle)}{\sum_{b=1}^k \exp(\langle \boldsymbol{w}_b,f(\boldsymbol{X})\rangle)}.
\end{align*}

\State Return the induced policy
\begin{align*}
\hat g(\boldsymbol{X})\coloneqq\arg\max_{a\in\mathcal{A}} \langle \boldsymbol{w}_a,\hat{f}_s(\boldsymbol{X})\rangle.
\end{align*}
\end{algorithmic}
\end{algorithm}

\section{Theoretical Properties}
\label{sec:theoretical_properties}
We now establish theoretical guarantees for the proposed learning framework. First, we establish that minimizing the surrogate loss \eqref{eq:loss-function} yields the optimal policy. We then derive a finite sample convergence rates via an oracle inequality that accounts for both the estimation of the embedding function $f$ and the estimation of the identification bounds.

\begin{restatable}{thm}{boundingexessloss}
\label{thm:bounding-excess-loss}
Let $\mathcal{G}$ be the class of all measurable functions $f: \mathcal{X} \to \mathbb{R}^{k-1}$. Assume that:
\begin{enumerate}
    \item For almost every $\boldsymbol{X} \in \mathcal{X}$, the minimax policy $g^*(\boldsymbol{X}) = \argmin_{a \in \mathcal{A}} \mathcal{L}(\boldsymbol{X}, a)$ is unique.
    \item The maximum loss is integrable: $\mathbb{E}[\max_{a \in \mathcal{A}} |\mathcal{L}(\boldsymbol{X}, a)|] < \infty$.
\end{enumerate}
Then, the infimum of the surrogate risk is the worst-case risk of the minimax policy:
\begin{align*}
\inf_{f \in \mathcal{G}} \mathcal{R}^s(f) = \mathcal{R}(g^*).
\end{align*}
Furthermore, for any $f \in \mathcal{G}$ the excess risk is bounded by the excess surrogate risk 
\begin{align*}
    0 \leq \mathcal{R}(g_f) - \mathcal{R}(g^*) \leq k \cdot \left( \mathcal{R}^s(f) - \inf_{f' \in \mathcal{G}} \mathcal{R}^s(f') \right),
 \end{align*}
where $k=\vert\mathcal{A}\vert$.
\end{restatable}

Theorem \ref{thm:bounding-excess-loss} justifies the use of $\mathcal{R}^s(f)$ as a valid objective: convergence to the minimum of the surrogate risk guarantees convergence to the optimal worst-case risk policy.

To decouple the estimation of the identification bounds from the optimization step, we define $\hat{f}_s$ using sample splitting. This independence is used in the proof of the oracle inequality to control the plug-in error introduced by $\hat{\mathcal{L}}$. 
One subsample is used to estimate $\hat{L}_a,\hat{U}_a$ (equivalently $\hat{\mathcal{L}}$), and the other to minimize the empirical surrogate risk. The oracle inequality below is stated for this sample splitting estimator.

To derive an oracle inequality for an arbitrary normed vector space of functions $\mathcal{F}$, in which the estimated function is compared to the best achievable element in the space, we make the following assumptions:

\begin{assumption}[Boundedness]
\label{asm: bounds-bounds}
There exists a constant $M > 0$ such that for all $a \in \mathcal{A}$ and almost all $\boldsymbol{X} \in \mathcal{X}$:

\begin{align*}    
|L_a(\boldsymbol{X})|, |U_a(\boldsymbol{X})|, |\hat{L}_a(\boldsymbol{X})|, |\hat{U}_a(\boldsymbol{X})|  \leq M.
\end{align*}
\end{assumption}

\begin{assumption}[Entropy of the function class]
\label{asm: entropy}
Let $\mathcal{B}_0 = \{f \in \mathcal{F}: \Vert f \Vert_{\mathcal{F}} \leq 1\}$ be the unit ball in $\mathcal{F}$. We assume there exists a constant $v \in (0, 2)$ such that the covering number of $\mathcal{B}_0$ satisfies:
\begin{align*}
    \sup_Q \log N(\epsilon, \mathcal{B}_0, L_2(Q)) \leq O(\epsilon^{-v}), \quad \text{for all } \epsilon > 0,
\end{align*}
where the supremum is taken over all discrete probability measures $Q$. Here, $N(\varepsilon, \mathcal{B}_0, L_2(Q))$ denotes the covering number, which is the minimal number of closed $\varepsilon$-balls in the $L_2(Q)$ norm required to cover $\mathcal{B}_0$
\end{assumption}

\begin{assumption}[Existence of minimizer]
There exists $f^*_\lambda \in \mathcal{F}$ such that
\label{asm: minimizer}
    $$\mathcal{R}^s(f^*_\lambda) + \lambda \Vert f^*_\lambda \Vert_\mathcal{F}^2 = \inf_{f \in \mathcal{F}} \mathcal{R}^s(f) + \lambda \Vert f\Vert_\mathcal{F}^2.
    $$ 
\end{assumption}

\begin{assumption}[Existence of empirical minimizer]
There exists $\hat{f}_s \in \mathcal{F}$ such that
\label{asm:emp-minimizer}
    $$\hat{\mathcal{R}}^s(\hat{f}_s) = \inf_{f \in \mathcal{F}} \hat{\mathcal{R}}^s(f), 
    $$ 
\end{assumption}

\begin{assumption}[Convergence rate of the bounds]
\label{asm: convergence}
Assume that for any $a \in \mathcal{A}$, the estimated bounds $\hat{L}_a$ and $\hat{U}_a$ satisfy
\begin{align*}
\mathbb{E}\left[ \left\vert \hat{L}_a(\boldsymbol{X}) -L_a(\boldsymbol{X})\right\vert \right] &= O(n^{-\alpha}), \\
\mathbb{E}\left[ \left\vert \hat{U}_a(\boldsymbol{X}) -U_a(\boldsymbol{X})\right\vert \right] &= O(n^{-\beta}),
\end{align*}
for some $\alpha, \beta > 0$
\end{assumption}

\begin{assumption}[Supremum norm bound]
\label{asm: norm}
Assume that for any $f \in \mathcal{F}$, there exists a constant $\kappa>0$ such that
$$\Vert f \Vert_\infty \leq \kappa \Vert f \Vert_\mathcal{F},
$$ where the supremum norm is defined as $\Vert f \Vert_{\infty}=\sup_{\boldsymbol{X}\in \mathcal{X}} \max_j(\vert f^j(\boldsymbol{X}) \vert)$.
\end{assumption}

Assumptions \ref{asm: entropy}, \ref{asm: minimizer}, and \ref{asm: norm} are standard regularity conditions in the analysis of kernel based methods (see Steinwart \& Christmann, \cite{Steinwart2008}, Ch. 6-7). Assumptions \ref{asm: bounds-bounds} and \ref{asm: convergence} are specific to the partial identification setting. They ensure that the estimation error of the identification bounds does not dominate the learning rate. Similar conditions have been used by \cite{puzhang2021}.

Before stating the theorem, define the regularization approximation error
\begin{align*}
a_{\mathcal{F}}(\lambda) \coloneqq \inf_{f\in\mathcal{F}}
\{\mathcal{R}^s(f)+\lambda\|f\|_{\mathcal{F}}^2\}-\inf_{f\in\mathcal{F}}\mathcal{R}^s(f).
\end{align*}
Note that $a_{\mathcal{F}}(\lambda) \to 0$ as $\lambda \to 0$, since a vanishing regularization term does not change the infimum of the surrogate risk.

\begin{restatable}{thm}{oracleinequality}
\label{thm:oracle-inequality}
    Let $\mathcal{F}$ be a vector space of measurable functions from $\mathcal{X} \to \mathbb{R}^{k-1}$, equipped with an inner product and its induced norm. Under \cref{asm: bounds-bounds}-\ref{asm: norm} and a regularizer satisfying $0<\lambda \leq 1$ we have:
    \begin{align*}
    \mathcal{R}^s(\hat{f}_s) - \inf_{f\in \mathcal{F}} \mathcal{R}^s(f)\leq a_{\mathcal{F}}(\lambda) +  O( (n\lambda)^{-\frac{1}{2}}+ n^{-\alpha}+n^{-\beta}).
    \end{align*}
\end{restatable}

If we define the approximation error as the gap between the best surrogate risk achievable over the class $\mathcal{F}$ and 
that achievable over the class of all measurable functions $\mathcal{G}$, namely 
\begin{align*}
\mathcal{A}(\mathcal{F})\coloneqq
\inf_{f\in\mathcal{F}}\mathcal{R}^s(f)- \inf_{f' \in \mathcal{G}}\mathcal{R}^s (f'),
\end{align*}
then, by combining \cref{thm:bounding-excess-loss} and \cref{thm:oracle-inequality}, we obtain the following corollary.

\begin{corollary}
\label{lem:class-approximation}
Under \cref{asm: bounds-bounds}-\ref{asm: norm}, we have
\begin{align*}
\mathcal{R}(g_{\hat{f}_s})
\leq
\mathcal{R} (g^*) + k \cdot \left(\mathcal{A}(\mathcal{F}) + a_{\mathcal{F}}(\lambda)
+ O\left((n\lambda)^{-\frac{1}{2}}+ n^{-\alpha}+n^{-\beta}\right)\right).
\end{align*}
\end{corollary}

\subsection{Verification of Assumptions}
\label{sec:verification}

To ensure the applicability of our theoretical results, we verify that Assumptions \ref{asm: bounds-bounds}--\ref{asm: norm} are satisfied under standard settings. Specifically, we consider the case where the hypothesis space is generated by the Gaussian radial basis function (RBF) kernel, a common choice in practice due to its universal approximation properties.

\begin{remark}
Suppose the outcome variable is bounded in $[K_0,K_1]$. Then a valid bounding constant $M$ is given by $M = \max(\lvert K_0 \rvert, \lvert K_1 \rvert)$. The same holds for the empirical plug-in bounds, and thus \cref{asm: bounds-bounds} is satisfied.

For the discrete problems handled by \verb|causaloptim|, the bounds are given by minima or maxima of linear expressions in observable probabilities. Therefore, as shown in \cite{puzhang2021} (Supplement D.2), convergence rates of the estimated probabilities transfer directly to the plug-in bounds, implying that \cref{asm: convergence} is satisfied.
\end{remark}

\begin{restatable}{ex}{gaussverification}
\label{ex:gaussian-verification}
Let $\mathcal{X} \subset \mathbb{R}^d$ be compact. Consider the Gaussian RBF kernel $k_{\text{RBF}}(\boldsymbol{X}, \boldsymbol{X}') = \exp(-\|\boldsymbol{X} - \boldsymbol{X}'\|^2 / 2\sigma^2)$. Let $\mathcal{H}$ denote the associated RKHS. We define the function space as the product Hilbert space $\mathcal{F} = \mathcal{H}^{k-1}$, equipped with the inner product $\langle f, g \rangle_{\mathcal{F}} = \sum_{j=1}^{k-1} \langle f^j, g^j \rangle_{\mathcal{H}}$.
Then assumptions \ref{asm: entropy}, \ref{asm: minimizer}, \ref{asm:emp-minimizer} and \ref{asm: norm} are satisfied.
\end{restatable}
See \cref{apx:proofs} for a formal argument. 

In settings with mixed data, categorical covariates are mapped to vectors via one-hot encoding and concatenated with the standardized continuous covariates. The resulting space still satisfies the assumptions.

\section{Numerical Experiments}
\label{sec:experiments}
We evaluate the finite sample performance of the proposed framework through four simulation studies. Each study varies a single parameter of the data generating process while holding the others fixed, thereby isolating the effect of (1) sample size, (2) unmeasured confounding strength, (3) instrumental variable strength, and (4) the number of treatment options. The code for the simulation as well as an implementation of the algorithm, can be found on \url{https://github.com/jhruza/Multicategory_Partial_Identification_Learning}.

\paragraph{Data generating process.}
We use an exact data generating process (ExactDGP), described in detail in the Supplementary Material (\cref{apx:simulation-setup}). The setup mimics an observational study with $d = 5$ continuous covariates $\boldsymbol{X} \sim \text{Uniform}(-1,1)^5$, $k = 3$ categorical treatment options, a instrumental variable $Z$ with $|Z| = 3$ levels, and a binary outcome $Y$. 

A discrete latent confounder $U \in \{-2, -1, 0, 1, 2\}$ with a bell-shaped distribution induces selection bias by influencing both the treatment assignment and the outcome. The treatment assignment follows a multinomial logistic model $P(A \mid \boldsymbol{X}, Z, U) \propto \exp(\boldsymbol{X}^\top \boldsymbol{\beta}_A + \boldsymbol{Z}^\top \boldsymbol{\gamma} \cdot \lambda_{\text{iv}} + \boldsymbol{\delta}_A \cdot U)$, where $\lambda_{\text{iv}}$ scales the instrument strength and $\boldsymbol{\delta}_A =(\delta_1, \dots, \delta_k)$ where $\delta_a = \lambda_{\text{conf}} \cdot (2(a-1) - k + 1)/(k-1)$ scales the strength of the confounder $U \to A$. Each treatment is assigned a treatment-specific factor evenly spaced between $-1$ and $1$, scaled by $\lambda_{\text{conf}}$, so that the latent confounder $U$ affects treatment assignment differently across treatment options. The outcome probability is $P(Y=1 \mid \boldsymbol{X}, A=a, U) = \sigma(\boldsymbol{X}^\top \boldsymbol{\beta}_{Y,a} + \lambda_{\text{conf}} \cdot U + \tau_a)$, where $\lambda_{\text{conf}}$ controls the confounding strength. Since $U$ is discrete with known support and distribution, we can analytically compute the potential outcomes (ground truth) $\mu_a(\boldsymbol{X}) = \sum_u P(Y=1 \mid \boldsymbol{X}, A=a, U=u) \cdot P(U=u)$ by marginalizing over the confounder. We use this to derive the oracle benchmark policy.

\paragraph{Use of true probabilities.}
A key design choice in our simulations is that we supply the \emph{true} conditional joint probabilities $P(A=a, Y=y \mid Z=z, \boldsymbol{X})$ to \verb|causaloptim| when computing the partial identification bounds, rather than estimating them from data. We do the same for the MOML method, by calculating the propensity scores $P(A=a \mid \boldsymbol{X}=\boldsymbol{x})$ using the ground truth. This removes the estimation step (as well as the sample splitting) and we can asses the performance of the policy learning algorithms. To offset this advantage, for these methods, we use a large training sample size of $n = 12{,}000$ (except in the sample size experiment), providing the outcome regression baselines with sufficient data to accurately learn models for $\mathbb{E}[Y \mid \boldsymbol{X}, A]$.

\paragraph{Comparators.}
We compare the following methods:
\begin{enumerate}
    \item \textbf{KerriskMin} (proposed): The kernelized risk minimizer from \cref{sec:kernelized-formulation}, using a Gaussian RBF kernel with bandwidth selected by the median heuristic and regularization $\lambda = 10^{-4}$.
    \item \textbf{NeuralriskMin} (proposed): The neural network parameterization from \cref{sec:neural-network}, using two hidden layers in a MLP with 64 nodes each and the same surrogate loss.
    \item \textbf{MOML}: The Multicategory Outcome-weighted Margin-based Learning method of Zhang et al.\ \cite{Zhang2020}, a point-identified ITR estimator that assumes no unmeasured confounding. We supply the propensity scores calculated from the data generating mechanism.
    \item \textbf{Tree OutcomeReg}: Outcome regression via a \texttt{GradientBoostingClassifier} from \texttt{scikit-learn} (100 trees, max depth 3). The model estimates $\hat{\mu}_a(\boldsymbol{X})$ for each treatment, and the ITR selects $\arg\max_a \hat{\mu}_a(\boldsymbol{X})$.
    \item \textbf{NeuralOutReg}: A neural network outcome regressor (two hidden layer MLP, 64 nodes each) that directly estimates $\hat{\mu}_a(\boldsymbol{X})$ and selects the treatment with the highest predicted outcome.
    \item \textbf{Minimax rule}: The minimax policy from \cref{prop:minimax-decision}, which selects $g^*(\boldsymbol{X}) = \arg\min_a \mathcal{L}(\boldsymbol{X}, a)$ using the true identification bounds. This is the best achievable policy under partial identification without any learning or smoothing.
\end{enumerate}

\paragraph{Evaluation metrics.}
Since we have access to the ground truth $\mu_a(\boldsymbol{X})$, we evaluate each method on an independent test set of $n_{\text{test}} = 5{,}000$ samples using three metrics: \textit{Policy agreement}, the proportion of test subjects that agree with the oracle optimal treatment $\arg\max_a \mu_a(\boldsymbol{X})$: \textit{oracle risk} (\cref{eq:oracle-risk}), $\mathcal{R}^*(g) = \mathbb{E}[\max_a \mu_a(\boldsymbol{X}) - \mu_{g(\boldsymbol{X})}(\boldsymbol{X})]$; and \textit{expected outcome}, $\mathbb{E}[\mu_{g(\boldsymbol{X})}(\boldsymbol{X})]$, the clinical value of the policy. Each simulation is repeated across at least $10$ Monte Carlo runs (see \cref{tab:dgp-parameters} in the supplementary material for detailed information), and we report means with bootstrapped confidence intervals.

\paragraph{Training and hyperparameter selection.}
We note that no systematic hyperparameter tuning was performed for any method. For the proposed kernelized method, we select the kernel bandwidth $\sigma$ via the median heuristic \cite{garreau2017large} (setting $\sigma$ to the median of pairwise Euclidean distances in the training data). The neural network variants (NeuralRisktMin and NeuralOutReg) use default learning rates. For the Tree OutcomeReg baseline, we use the default configuration with $100$ trees and maximum depth $3$. The MOML baseline uses the same kernel and regularization settings as KerRiskMin. The consistent performance of the proposed methods across the different settings suggests that the results are not simply due to fine tuned models.

\paragraph{Default parameters.}
Unless otherwise stated, the default configuration is $n = 12{,}000$ training samples, confounding strength $\lambda_{\text{conf}} = 4$, IV strength $\lambda_{\text{iv}} = 1$, $k = 3$ treatments, and $|Z| = 3$ instrument levels.

\subsection{Finite Sample Convergence}
\label{sec:sim-sample-size}

We first examine convergence behavior in a favorable setting of no unmeasured confounding ($\lambda_{\text{conf}} = 0$) as a sanity check. We vary the training sample size $n$ from $5$ to $20{,}000$ while holding all other parameters at their defaults.

\Cref{fig:sample_size} presents the results. In the setting with no unmeasured confounding, the flexible outcome regression baseline NeuralOutReg achieves the highest policy agreement and lowest oracle risk for large $n$. This is expected: without confounding, $\mu_a(\boldsymbol{X}) = \mathbb{E}[Y \mid \boldsymbol{X}, A=a]$ is directly identifiable from the observed data, and outcome regression models this quantity directly. The proposed partial identification methods (KerRiskMin and NeuralRiskMin) converge to the performance of the minimax rule, which they target. The minimax policy is suboptimal relative to the oracle when confounding is absent, as minimax policy is ``unaware'' that there is no unmeasured confounding and accounts for it nevertheless.

\begin{figure}[ht]
    \centering
    \includegraphics[width=\linewidth]{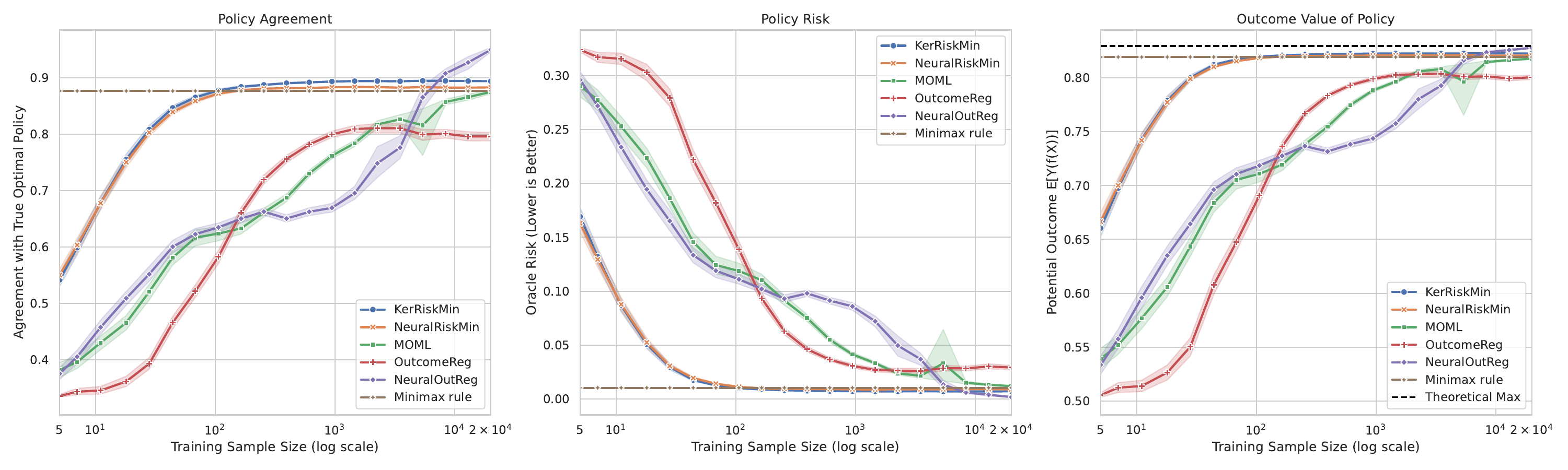}
    \caption{Policy agreement, oracle risk, and expected outcome as a function of training sample size $n$, with $\lambda_{\text{conf}} = 0$ (no confounding). All methods converge as $n$ increases. In the absence of confounding, the outcome regression method NeuralOutReg achieves the highest agreement and lowest oracle risk, consistent with the fact that they directly estimate the identifiable $\mu_a(\boldsymbol{X})$. The proposed methods (KerRiskMin and NeuralRiskMin) converge to the minimax rule and achieve good performance even though they are optimizing a conservative objective.}
    \label{fig:sample_size}
\end{figure}

This experiment establishes a baseline: in the best case for the comparators, outcome regression is the strongest approach.

\subsection{Robustness to Unmeasured Confounding}
\label{sec:sim-confounding}

Based on the good performance of the outcome regression models of previous simulation we now fix $n = 12{,}000$ and increase the confounding strength $\lambda_{\text{conf}}$ from 0 to 10.

\Cref{fig:confounding} shows a shift in performance. At $\lambda_{\text{conf}} = 0$ (also implying $\boldsymbol{\delta}_A=\boldsymbol{0}$), the outcome regression methods (Tree OutcomeReg, NeuralOutReg) and the MOML baseline achieve approximately $78 - 93\%$ agreement with the true optimal policy. However, as confounding increases, their performance degrades substantially: by $\lambda_{\text{conf}} = 10$, Tree OutcomeReg drops to approximately $57\%$ agreement and NeuralOutReg as well as MOML to around $70\%$. These methods rely on the assumption no unmeasured confounding assumption which is gradually violated as $\lambda_{\text{conf}}$ increases.

\begin{figure}[ht]
    \centering
    \includegraphics[width=\linewidth]{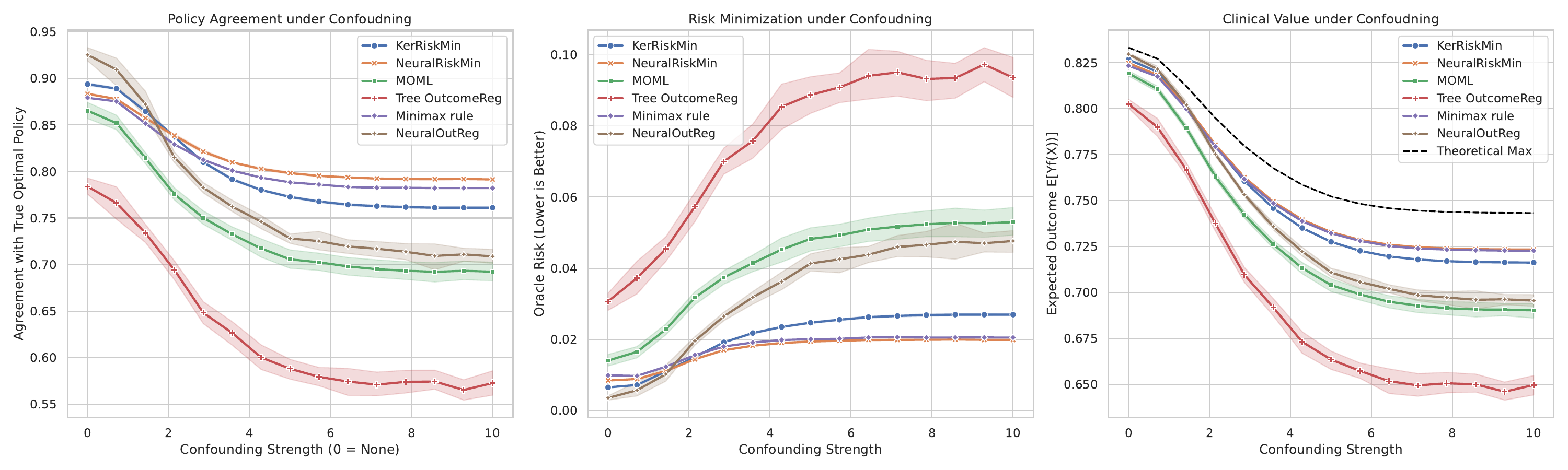}
    \caption{Performance as a function of confounding strength $\lambda_{\text{conf}}$, with $n = 12{,}000$. Outcome regression methods degrade rapidly as confounding increases, while the proposed partial identification methods (KerRiskMin and NeuralRiskMin) remain stable and close to the minimax rule. Note that even the oracle minimax rule shows a gradual decline, reflecting the widening of the identification bounds under stronger confounding.}
    \label{fig:confounding}
\end{figure}

In contrast, the proposed methods (KerRiskMin and NeuralRiskMin) degrade much slower and for $\lambda_{\text{conf}}>2$ are superior than the other methods. The partial identification bounds, derived from the instrumental variable $Z$, remain valid regardless of confounding strength. The gradual decline in absolute performance reflects the widening of the bounds as the confounding increases.  Stronger confounding increases identification uncertainty, but the bounds never become invalid.

MOML uses a causal framework (inverse probability weighting), but it also assumes all confounders are measured and thus also degrades. Beyond $\lambda_{\text{conf}} \approx 6$, all methods plateau. This happens since the confounding term $\lambda_{\text{conf}} \cdot U$ begins to dominate the logistic link function, so that the outcome is increasingly determined by $U$ alone.

\subsection{Impact of Instrument Strength}
\label{sec:sim-iv-strength}

\begin{figure}[ht]
    \centering
    \includegraphics[width=\linewidth]{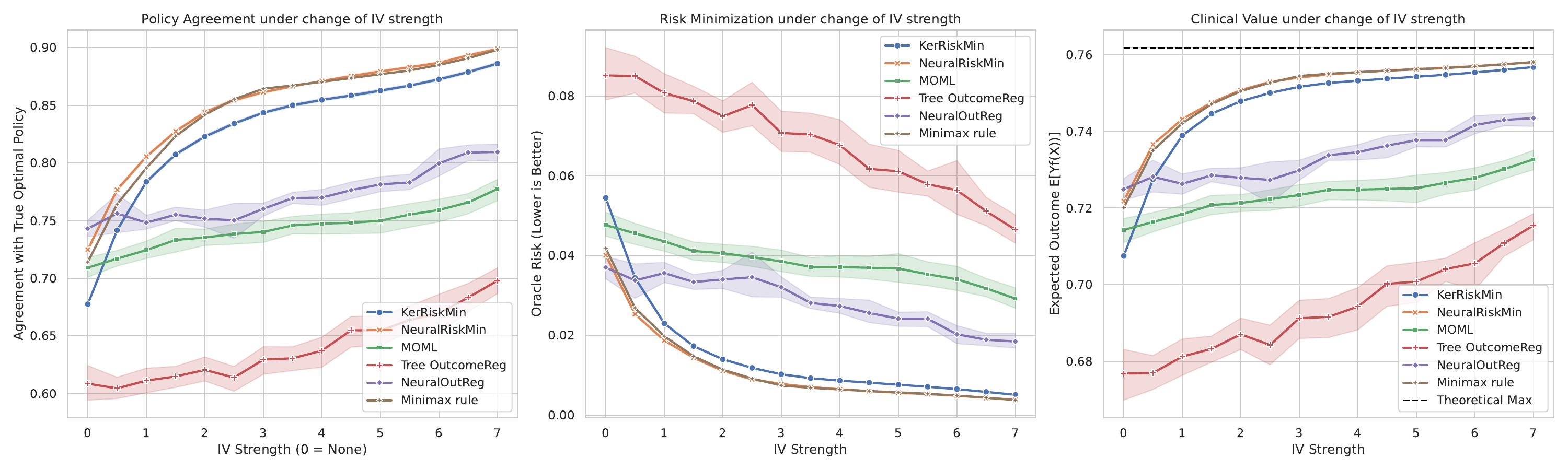}
    \caption{Performance as a function of instrumental variable strength $\lambda_{\text{iv}}$, with $n = 12{,}000$ and $\lambda_{\text{conf}} = 4$. The partial identification methods improve with instrument strength as the bounds narrow. The outcome regression baselines are less sensitive to the instrument strength, since they do not use $Z$ in the models.}
    \label{fig:iv_strength}
\end{figure}

This simulation investigates how the information content of the instrument affects the learned policy. We fix the confounding at $\lambda_{\text{conf}} = 4$ and vary $\lambda_{\text{iv}}$ from 0 (instrument provides no information about treatment assignment) to 7 (strong instrument). In the partial identification framework, instrument strength is directly related to bound width. A weak instrument yields wide, uninformative intervals $[L_a(\boldsymbol{X}), U_a(\boldsymbol{X})]$, while a strong instrument tightens these.

\Cref{fig:iv_strength} shows a clear monotonic relationship between instrument strength and the performance of the partial identification methods. When $\lambda_{\text{iv}} = 0$, the instrument provides no leverage and the bounds are wide; our methods perform not better than the uninformed baseline. As $\lambda_{\text{iv}}$ increases, the bounds tighten progressively. For strong instruments ($\lambda_{\text{iv}} \geq 5$), KerRiskMin and NeuralRiskMin converge to the oracle policy.

The outcome regression baselines (Tree OutcomeReg, NeuralOutReg) and MOML, maybe at first surprising, increase a small amount in performance. Although the outcome regression does not explicitly adjust for the instrument, a stronger instrument decreases the share of treatment variation driven by unobserved factors, thereby reducing the relative influence of the unobserved confounder on treatment assignment. As a result, the observed association $\mathbb{E}[Y \mid \boldsymbol{X},A=a]$ gets closer to the causal quantity $\mu_a(\boldsymbol{X})$ leading to a modest reduction in bias even without directly making use of the instrument variable.
Nevertheless their performance remains bounded at a level the confounding allows, regardless of the information available in the data.

This is in line with \cref{prop:risk-error}: the gap between the oracle risk and the worst-case risk is controlled by the width of the identification bounds, and tighter bounds translate directly into better policies.

\subsection{Scalability to Multiple Treatments}
\label{sec:sim-treatments}

Lastly, we examine the scalability of the simplex embedding as the number of treatments $k$ increases from 2 to 6, while the instrument cardinality is held fixed at $|Z| = 3$. 

\Cref{fig:treatment_change} shows that our proposed methods maintain strong performance for $k \leq 3$. However, a sharp decline occurs when $k$ exceeds $|Z| = 3$. For $k = 4$, agreement drops drastically for the proposed methods, and by $k = 6$ the performance has degradated completely. This degradation is not a limitation of the learning algorithm but rather reflects underlying mechanism of partial identification. The instrumental variable $Z$ with $|Z|$ levels can induce at most $|Z|$ distinct conditional treatment distributions $P(A \mid Z = z, \boldsymbol{X})$. When the number of treatments exceeds the number of instrument levels, there are insufficient constraints to tightly bound all $k$ potential outcomes simultaneously, and the identification intervals widen substantially. As wider bounds increase the worst-case risk for every action, this makes the minimax policy less likely to agree with the oracle policy, as it is conservative.

\begin{figure}[ht]
    \centering
    \includegraphics[width=\linewidth]{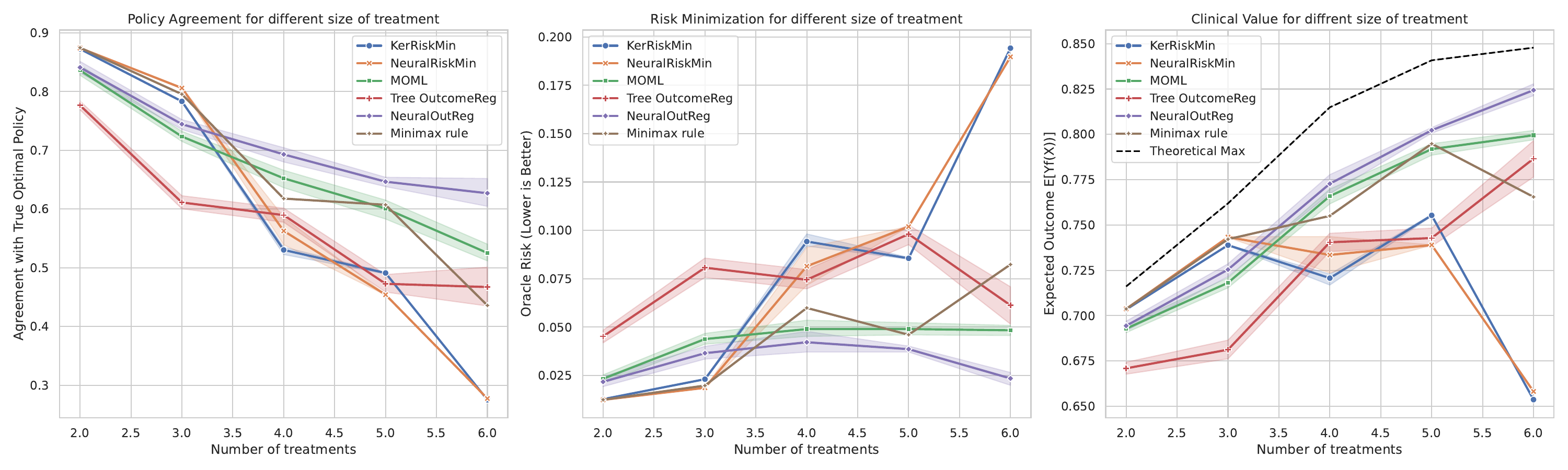}
    \caption{Performance as the number of treatments $k$ increases from 2 to 6, with $|Z| = 3$ instrument levels held fixed. All methods degrade as $k$ grows, but a large drop occurs when $k$ exceeds $|Z| = 3$. This reflects a structural constraint: the instrument can only induce $|Z|$ distinct treatment distributions, and when $k > |Z|$, the identification bounds widen substantially.}
    \label{fig:treatment_change}
\end{figure}

It is worth noting that the outcome regression baselines also degrade as $k$ increases, even though they do not use the instrument. This suggests increasing difficulty of the decision problem itself, as with more treatments, the probability that the different treatment can be close in outcome increases. All methods require more information to make the correct decisions. However, the partial identification methods exhibit a sharper decline after $k>3$.
This suggests that when applying this method in practice, the instrument should have sufficient number of levels relative to the number of treatments.

\section{Discussion}
\label{sec:discussion}
We introduced a partial identification learning framework for individualized treatment rules with categorical treatments, outcomes, and instrumental variables. The main contribution of the paper is to extend partial identification based policy learning beyond the binary setting. This is relevant in applications where treatment choice is multicategorical and where point identification of causal effects is difficult to justify, such as observational studies with unmeasured confounding or randomized trials with noncompliance. 

Our experiments highlight the tradeoffs of the proposed framework. In settings without unmeasured confounding, point identified methods such as outcome-regression-based models can outperform the proposed approach, since they target the fully identifiable oracle rule directly. In contrast, our method is designed for robustness when such assumptions are not credible. This robustness, however, depends strongly on the informativeness of the partial identification bounds. When the bounds are narrow, the minimax policy is close to the oracle policy and the learned rule performs well. When the bounds are wide, the resulting rule becomes more conservative. The theory and simulations show that stronger instruments improve performance because they narrow the bounds, while weak instruments lead to broader identification regions and less informative decisions. They also show a structural limitation of the partial identification problem. When the instrument has fewer effective levels than the treatment, the available variation may be insufficient to constrain the potential outcomes, causing the bounds-based methods to deteriorate rapidly.

While we focus on the minimax criterion, the same general learning idea can be adapted to other deterministic decision criteria under partial identification. In particular, one may consider criteria obtained by extending the binary rules discussed by \textcite{Cui2020} to the multicategory setting. Whenever such a criterion can be represented through a pointwise criterion specific loss derived from the bounds, the corresponding decision rule can be written as choosing the action that minimizes this loss. In such cases, the current worst-case loss can in principle be replaced by another criterion specific loss. For example, a multicategory extension of the pessimistic criterion would correspond to choosing the treatment with the largest lower bound. By contrast, extensions based on randomized decision rules would require additional methodological development, since the target object is no longer a deterministic policy.

From a computational perspective, the proposed surrogate risk makes gradient based optimization possible, but the surrogate risk is generally not convex. Consequently,  performance may depend on initialization, optimization strategy, and regularization choices. In the present paper, our goal was to establish the framework rather than to optimize all implementation details. In practice, however, parameters such as the regularization factor $\lambda$, the kernel bandwidth $\sigma$ and other model parameters may be selected in a data-adaptive way using a validation set or cross validation \parencite{Arlot2010}. 

Several limitations should be noted. First, the framework inherits the assumptions required by the method used to derive the identification bounds, so validity of the learned policy depends on validity of the underlying causal model. Second, even valid bounds may be too wide to support useful treatment decisions for example if an instrument is weak or otherwise uninformative. Finally, the present paper focuses on categorical treatments, outcomes, and instruments, and therefore does not yet address settings with continuous decision or outcome spaces. One could also consider settings in which treatment options are not fully distinct, but instead have an inherent structure such as an ordering. In these cases, the simplex embedding may be unnecessary. This is relevant, for example, when treatments correspond to different doses of the same drug rather than to completely different drugs.

Natural extensions include combining the framework with other sources of partial identification, such as proxy-variable approaches, and extending it to ordinal or continuous treatments and outcomes as well as dynamic treatment regimes. 
\printbibliography 
\begin{refsection}
\appendix

\section{Additional Results} \label{apx:additional-results}

\begin{lemma}
\label{lemma: probability_bound}
Let the decision rule be defined as $g_f(\boldsymbol{X}) := \arg\max_{a \in \mathcal{A}} \langle \boldsymbol{w}_a, f(\boldsymbol{X}) \rangle$. For any measurable function $f$ and covariates $\boldsymbol{X}$, the softmax-weight $p_{g_f(\boldsymbol{X})}(\boldsymbol{X}, f)$ satisfies the lower bound:
\begin{align*}    
    p_{g_f(\boldsymbol{X})}(\boldsymbol{X}, f) \geq \frac{1}{k}.
\end{align*}
\end{lemma}

\begin{lemma}\label{lemma: bound-risk-diff}
   Under Assumption \ref{asm: convergence}, for any fixed $f \in \mathcal{F}$ we have
    \begin{align*}
        \mathbb{E}[  \vert \mathcal{L}^s(\boldsymbol{X}, f) - \hat{\mathcal{L}}^s(\boldsymbol{X}, f)\vert ] \leq O(n^{-\alpha}+n^{-\beta}).
    \end{align*}
\end{lemma}

\begin{lemma}
    \label{lemma: lipschitz-covering}
Let $\Phi\colon \mathcal{B}^* \to \mathcal{F}_{\mathcal{L}}$ defined by $\Phi(f) := \mathcal{L}^s(\cdot, f)$ be surjective and Lipschitz with constant $L_{lip}$ then,
\begin{align*}
N(\epsilon, \mathcal{F}_{\mathcal{L}}, \Vert \cdot \Vert_{\mathcal{F}_{\mathcal{L}}}) \leq N\left(\frac{\epsilon}{L_{lip}} , \mathcal{B}^*, \Vert \cdot \Vert_{\mathcal{B}^*}\right).
\end{align*}
\end{lemma}

\begin{lemma}
\label{lemma: bound-empirical}
Under the assumptions of \cref{thm:oracle-inequality} we have 
\begin{align*}
    \mathbb{E}\left[ \mathbb{E}[\mathcal{L}^s(\boldsymbol{X},\hat{f}_s)] - \frac{1}{\vert I_2\vert}\sum_{i\in I_2}\mathcal{L}^s(\boldsymbol{X}_i, \hat{f}_s)\right] \leq O\left(\frac{1}{\sqrt{n \lambda}}\right).
\end{align*}
\end{lemma}

\section{Proofs}\label{apx:proofs}

\begin{proof} [Proof of \cref{prop:minimax-decision}]
To show that $g^*$ obtains the infimum, we need to show that every other function of $\mathcal{D}$ yields a larger or equal risk and that $g^*$ is measurable.
By definition of the min we have for any $a \in \mathcal{A}$ and for any $\boldsymbol{X}\in \mathcal{X}$:
\begin{align*}
    \mathcal{L}(\boldsymbol{X}, g^*(\boldsymbol{X}))\leq \mathcal{L}(\boldsymbol{X}, a)
\end{align*}
This implies also for any $g \in \mathcal{D}$.
\begin{align*}
    \mathcal{L}(\boldsymbol{X}, g^*(\boldsymbol{X}))\leq \mathcal{L}(\boldsymbol{X}, g(\boldsymbol{X}))
\end{align*}
We take the expectation on both sides and by monotonicity of the expectation,
\begin{align*}
    \mathcal{R}(g^*) \leq \mathcal{R}(g).
\end{align*}
Finally, since the action set $\mathcal{A}$ is finite, the pointwise minimizer is measurable. In particular, a deterministic tie-breaking rule (e.g., selecting the action with the smallest index) ensures that the $\argmin$ is uniquely defined and measurable. Therefore  $g^* \in \mathcal{D}$.
\end{proof}

\begin{proof} [Proof of claim in \cref{rmk: CATE-connection}]
By assumption our action space consists of two elements $\mathcal{A}=\{-1,1\}$. To show the equivalence we use the definition of the CATE from \cite{puzhang2021}, let $\text{CATE}(\boldsymbol{X}):=\mathbb{E}[Y(1)-Y(-1) \mid \boldsymbol{X}]$ and $L_{\text{CATE}}(\boldsymbol{X}), U_{\text{CATE}}(\boldsymbol{X})$ be the partial identification bounds of the CATE, see \cref{sec:partial-identification}.
We proceed by showing that the worst-case loss equals the pointwise counterpart of the "risk function" $\mathcal{R}_{\text{upper}}$ defined in Definition 1 of \cite{puzhang2021}, since taking expectation over both sides finishes the proof. Formally we have to show that:
\begin{align}
\label{eq:regret-CATE}
    \mathcal{L}(\boldsymbol{X},a)^+ = \sup_{C'(\boldsymbol{X})\in [L_{\text{CATE}}(\boldsymbol{X}), U_{\text{CATE}}(\boldsymbol{X})]} |C'(\boldsymbol{X})| \mathds{1} \{\text{sgn}(C'(\boldsymbol{X})) \neq \text{sgn}(a)\}
\end{align}
where $\text{sgn}(x) = 1, \quad \forall x>0, x \in \mathbb{R}$ and $-1$ otherwise.

By variation independence and the definition of the CATE in \cite{puzhang2021} the partial identification bounds of the CATE can be decomposed to:
    \begin{align*}
        L_{\text{CATE}}(\boldsymbol{X})=L_1(\boldsymbol{X}) - U_{-1}(\boldsymbol{X}) \\
        U_{\text{CATE}}(\boldsymbol{X})=U_1(\boldsymbol{X}) - L_{-1}(\boldsymbol{X})
    \end{align*}
We are left to check the equality of \cref{eq:regret-CATE} for each case. 
\\ \\
\underline{If $a=1$} then $\mathcal{L}(\boldsymbol{X},1)=U_{-1}(\boldsymbol{X})-L_1(\boldsymbol{X}) = -L_{\text{CATE}}(\boldsymbol{X})$\\
    \par If $0\leq L_{\text{CATE}}(\boldsymbol{X})$, then we also have $0\leq L_{\text{CATE}}(\boldsymbol{X}) \leq U_{\text{CATE}}(\boldsymbol{X})$ this implies that $C'(\boldsymbol{X}) \geq 0$ hence $\text{sgn}(C'(\boldsymbol{X}))=1=a$ and the RHS of \cref{eq:regret-CATE} equals zero. Also $\mathcal{L}(\boldsymbol{X},1) < 0 \implies \mathcal{L}(\boldsymbol{X},1)^+=0$.\\
    \par If $0 > L_{\text{CATE}}(\boldsymbol{X})$ then $\mathds{1} \{\text{sgn}(C'(\boldsymbol{X})) \neq 1\}=1$ for $C'(\boldsymbol{X})\leq 0$ and hence $\sup |C'(\boldsymbol{X})| \mathds{1} \{\text{sgn}(C'(\boldsymbol{X})) \neq 1\} = - L_{\text{CATE}}(\boldsymbol{X})= \mathcal{L}(\boldsymbol{X},1)$ by definition.\\ \\
\underline{Otherwise $a=-1$} then $\mathcal{L}(\boldsymbol{X},-1)=U_{1}(\boldsymbol{X})-L_{-1}(\boldsymbol{X}) = U_{\text{CATE}}(\boldsymbol{X})$\\
     \par If $0 \leq U_{\text{CATE}}(\boldsymbol{X})$ then $\mathds{1} \{\text{sgn}(C'(\boldsymbol{X})) \neq -1)=1$ for $C'(\boldsymbol{X})\geq 0$ and hence $\sup |C'(\boldsymbol{X})| \mathds{1} \{\text{sgn}(C'(\boldsymbol{X})) \neq -1\} = U_{\text{CATE}}(\boldsymbol{X})= \mathcal{L}(\boldsymbol{X},-1)$\\
     \par If $0 \geq U_{\text{CATE}}(\boldsymbol{X})$ then we also have $L_{\text{CATE}}(\boldsymbol{X}) \leq U_{\text{CATE}}(\boldsymbol{X}) \leq 0$  and hence the RHS of \cref{eq:regret-CATE} is zero as $\text{sgn}(C'(\boldsymbol{X}))=-1, \quad \forall C'(\boldsymbol{X})\in [L_{\text{CATE}}(\boldsymbol{X}), U_{\text{CATE}}(\boldsymbol{X})]$. But also $\mathcal{L}(\boldsymbol{X},-1)= U_{\text{CATE}}(\boldsymbol{X}) \leq 0$ implying $\mathcal{L}(\boldsymbol{X},-1)^+=0$
\end{proof}

\begin{proof}[Proof of \cref{prop:risk-error}]
For any $\boldsymbol{X} \in \mathcal{X}$,
\begin{align*}
\mathcal{L}(\boldsymbol{X},g(\boldsymbol{X})) - \mathcal{L}^*(\boldsymbol{X},g(\boldsymbol{X}))&=  \mathcal{L}(\boldsymbol{X},g(\boldsymbol{X})) - \left[\max_a \mu_a(\boldsymbol{X}) - \mu_{g(\boldsymbol{X})}(\boldsymbol{X})\right]\\
&= \left[\max_{b \in \mathcal{A} \setminus \{g(\boldsymbol{X})\}} U_b(\boldsymbol{X})-\max_a \mu_a(\boldsymbol{X})\right] + [\mu_{g(\boldsymbol{X})}(\boldsymbol{X})-L_{g(\boldsymbol{X})}(\boldsymbol{X})]\\
&\leq \max_{b \in \mathcal{A} \setminus \{g(\boldsymbol{X})  \}}[U_b(\boldsymbol{X})-\mu_b(\boldsymbol{X})] + [\mu_{g(\boldsymbol{X})}(\boldsymbol{X})-L_{g(\boldsymbol{X})}(\boldsymbol{X})] \\
&\leq \max_{b \in \mathcal{A} \setminus \{g(\boldsymbol{X})  \}} W(b,\boldsymbol{X})+W(g(\boldsymbol{X}),\boldsymbol{X}) \\
&\leq \sum_a W(a,\boldsymbol{X}).
\end{align*}
This chain of pointwise inequalities shows that the difference between the worst-case loss and the oracle loss is bounded at any point $\boldsymbol{X}$. Specifically, it is less than or equal to the sum of all individual bound widths for a given $\boldsymbol{X}$.
Taking expectations over $\boldsymbol{X}$ on both sides completes the proof.
\end{proof}

\begin{proof} [Proof of \cref{thm:bounding-excess-loss}]
Recall the risk of the minimax policy (\cref{prop:minimax-decision}):
\begin{align*}
\mathcal{R}(g^*) = \mathbb{E}[\min_{a\in \mathcal{A}} \mathcal{L}(\boldsymbol{X}, a)].
\end{align*}
We first show $\mathcal{R}(g^*) = \inf_{f \in \mathcal{G}} \mathcal{R}^s(f)$ by constructing a converging series in $\mathcal{G}$ so its surrogate risk converges to $\mathcal{R}(g^*)$.
To construct this series, consider the space of weight vectors whose components sum to $1$, that is
\begin{align*}
    \Delta^{k}:=\{\boldsymbol{p}\in \mathbb{R}_{\geq0}^k \colon  \| \boldsymbol{p} \|_1 =1\}.
\end{align*}
Notice that the softmax weights given by \cref{eq:softmax-weights} lie in this space, i.e. $\forall f \in \mathcal{G}, \forall \boldsymbol{X} \in \mathcal{X}\colon \quad  \boldsymbol{p}(\boldsymbol{X},f) \in \Delta^k$.
By \cref{eq:loss-function}, the pointwise surrogate loss is $\mathcal{L}^s(\boldsymbol{X}, f) = \sum_{a=1}^k p_a(\boldsymbol{X},f)\mathcal{L}(\boldsymbol{X},a)$, a convex combination of worst-case losses. Since $\boldsymbol{p}(\boldsymbol{X},f) \in \Delta^k$, we have $\mathcal{L}^s(\boldsymbol{X}, f) \geq \min_{a \in \mathcal{A}} \mathcal{L}(\boldsymbol{X},a)$ for all $f$, with equality when the softmax concentrates entirely on $g^*(\boldsymbol{X})$, i.e.\ $\boldsymbol{p}(\boldsymbol{X},f) = \boldsymbol{e}_{g^*(\boldsymbol{X})}$, where $\boldsymbol{e}_{j}$ is the $j$-th standard basis vector in $\mathbb{R}^k$ (uniqueness follows from assumption 1).

The minimax policy $g^*(\boldsymbol{X})$ is a measurable function, as well as the map $\phi: \mathcal{A} \to \mathbb{R}^{k-1}$ defined in \eqref{eq: immersion}. Then $f_{g^*} := \phi \circ g^*$ is also measurable and thus in $\mathcal{G}$.
Consider the sequence $f_m(\boldsymbol{X}) := m \cdot f_{g^*}(\boldsymbol{X})$ for $m  \in \mathbb{N}$. This sequence is measurable and also in $\mathcal{G}$.
As $m \to \infty$, the softmax vector converges to:
\begin{align}
p_{g^*(X)}(\boldsymbol{X}, f_m) = \frac{\exp(m) }{ \exp(m) + (k-1) \exp(\frac{-m}{k-1})} \to 1,
\end{align}
where we have used the explicitly calculated inner product:
\begin{align*}
    \langle \boldsymbol{w}_a, \boldsymbol{w}_b \rangle =
    \begin{cases}
        1 & \text{if } a=b, \\
        \frac{-1}{k-1} & \text{if } a \neq b .
    \end{cases}
\end{align*}

Consequently $p_a(\boldsymbol{X}, f_m) \to 0$ for $a \neq g^*(\boldsymbol{X})$ as they have to sum to one. 
This shows $\boldsymbol{p}(\boldsymbol{X}, f_m) \to \boldsymbol{e}_{g^*(\boldsymbol{X})}$, where we have used assumption 1 to ensure the limit is unique.
Thus, the pointwise surrogate loss converges: $\lim_{m \to \infty} \mathcal{L}^s(\boldsymbol{X}, f_m) = \min_a \mathcal{L}(\boldsymbol{X}, a)$.
By Assumption 2 and the Dominated Convergence Theorem, we can interchange limit and expectation:
\begin{align*}
\lim_{m \to \infty} \mathcal{R}^s(f_m) = \mathbb{E}\left[\lim_{m \to \infty} \mathcal{L}^s(\boldsymbol{X}, f_m)\right] = \mathbb{E}[\min_{a \in \mathcal{A}} \mathcal{L}(\boldsymbol{X}, a)] = \mathcal{R}(g^*).
\end{align*}
Since $\mathcal{R}^s(f) \geq \mathcal{R}(g^*)$ for all $f$, this sequence shows that $\inf_{f \in \mathcal{G}} \mathcal{R}^s(f) = \mathcal{R}(g^*)$.

Second we establish the inequality in the theorem. 
For all $\boldsymbol{X}$, $a$ and all measurable functions $f$ by \cref{lemma: probability_bound} the term $p_{g_f(\boldsymbol{X})}(\boldsymbol{X}, f)$ can be bounded below by $\frac{1}{k}$. 

Let $f \in \mathcal{G}$ be arbitrary we can bound point wise
\begin{align*}
    \mathcal{L}(\boldsymbol{X},g_f(\boldsymbol{X}))-\mathcal{L}(\boldsymbol{X}, g^*(\boldsymbol{X}))
    &\leq k \cdot p_{g_f(\boldsymbol{X})} (\boldsymbol{X},f) \cdot (\mathcal{L}(\boldsymbol{X},g_f(\boldsymbol{X}))-\mathcal{L}(\boldsymbol{X}, g^*(\boldsymbol{X})))\\
    &\leq k \cdot \sum_{a=1}^k p_a(\boldsymbol{X},f)\cdot (\mathcal{L}(\boldsymbol{X},a)-\mathcal{L}(\boldsymbol{X}, g^*(\boldsymbol{X})))\\
    &=k  \cdot  \left( \sum_{a=1}^k p_a(\boldsymbol{X},f)\cdot \mathcal{L}(\boldsymbol{X},a)-\mathcal{L}(\boldsymbol{X}, g^*(\boldsymbol{X}))\right)
\end{align*}
Taking expectation yields
\begin{align*}
    \mathcal{R}(g_f)- \mathcal{R}(g^*)&=\mathbb{E}[\mathcal{L}(\boldsymbol{X},g_f(\boldsymbol{X}))-\mathcal{L}(\boldsymbol{X}, g^*(\boldsymbol{X}))]\\
    &\leq k  \cdot \mathbb{E}\left[ \sum_{a=1}^k p_a(\boldsymbol{X},f)\cdot \mathcal{L}(\boldsymbol{X},a)-\mathcal{L}(\boldsymbol{X}, g^*(\boldsymbol{X}))\right]\\
    &=k \cdot ( \mathcal{R}^s(f) -\mathcal{R}(g^*))\\
    &= k \cdot (\mathcal{R}^s(f)-\inf_{f'\in \mathcal{G}}\mathcal{R}^s(f'))
\end{align*}
where we have use the first result of the theorem in the last line to replace $\mathcal{R}(g^*)$. 
This completes the proof.
\end{proof}

\begin{proof}[Proof of \cref{lemma: probability_bound}]
Let $a^* := g_f(\boldsymbol{X})$ denote the action selected by the decision rule. By definition, $a^*$ maximizes the inner product, implying $\langle \boldsymbol{w}_{a^*}, f(\boldsymbol{X}) \rangle \geq \langle \boldsymbol{w}_b, f(\boldsymbol{X}) \rangle$ for all $b \in \mathcal{A}$.

Since the exponential function is monotonically increasing, we have $\exp(\langle \boldsymbol{w}_{a^*}, f(\boldsymbol{X}) \rangle) \geq \exp(\langle \boldsymbol{w}_b, f(\boldsymbol{X}) \rangle)$ for all $b$.
By a direct calculation:
\begin{align*}
    p_{a^*}(\boldsymbol{X}, f) &= \frac{\exp(\langle \boldsymbol{w}_{a^*}, f(\boldsymbol{X}) \rangle)}{\sum_{b=1}^k \exp(\langle \boldsymbol{w}_b, f(\boldsymbol{X}) \rangle)} \\
    &\geq \frac{\exp(\langle \boldsymbol{w}_{a^*}, f(\boldsymbol{X}) \rangle)}{\sum_{b=1}^k \exp(\langle \boldsymbol{w}_{a^*}, f(\boldsymbol{X}) \rangle)} \\
    &= \frac{\exp(\langle \boldsymbol{w}_{a^*}, f(\boldsymbol{X}) \rangle)}{k \cdot \exp(\langle \boldsymbol{w}_{a^*}, f(\boldsymbol{X}) \rangle)} \\
    &= \frac{1}{k}.
\end{align*}
\end{proof}

\oracleinequality*

\begin{proof}[Proof of \cref{thm:oracle-inequality}]
Through out the proof, if we do not specify the norm $\Vert \cdot \Vert$ we implicitly refer to the norm of the space $\mathcal{F}$ that is $\Vert\cdot \Vert_\mathcal{F}$.

The idea of the proof is to decompose $\quad \mathcal{R}^s(\hat{f}_s)-\inf_{f\in \mathcal{F}}\mathcal{R}^s(f)$ into several terms accounting for a small change between them. We then bound these terms desperately. Before we start, let us state some technical details.
To decouple the estimation of the worst-case loss from the optimization step, we use sample splitting. Split the index set $I=\{1,\dots,n\}$ into two disjoint parts $I_1 \sqcup I_2 = I$. Let $O_i:=(\boldsymbol{X}_i,A_i,Y_i,Z_i)$ denote the full observed data for unit $i$, and define the sigma-algebras
\begin{align*}
    \mathcal{G}_1:=\sigma(O_i:i\in I_1), \qquad \mathcal{G}_2:=\sigma(O_i:i\in I_2).
\end{align*}
Due to the split and that the observations are i.i.d., $\mathcal{G}_1$ and $\mathcal{G}_2$ are independent. Throughout the proof, $\boldsymbol{X}$ denotes a generic covariate draw with the same distribution as each $\boldsymbol{X}_i$. We use $I_1$ to estimate the worst-case loss $\hat{\mathcal{L}}$ and $I_2$ to minimize the empirical surrogate risk \cref{eq:emp-loss-function}. Denote $\hat{f}_s$ as the solution where we have performed sample splitting. In particular due to \cref{asm:emp-minimizer} we can write $\hat{f}_s$ as the $\argmin$ instead of the $\inf$:
\begin{align*}
\hat{f}_s:=\argmin_{f\in \mathcal{F}} \hat{\mathcal{R}}^s (f)=\argmin_{f\in \mathcal{F}} \frac{1}{|I_2|}\sum_{i \in I_2}\hat{\mathcal{L}}^s(\boldsymbol{X}_i, f) + \lambda \Vert f \Vert ^2,
\end{align*}
where $\mathcal{L}^s(\boldsymbol{X}, f) = \sum_{a=1}^k p_a(\boldsymbol{X},f)\mathcal{L}(\boldsymbol{X},a)$ is the surrogate loss (\cref{eq:loss-function}) and $\hat{\mathcal{L}}^s(\boldsymbol{X}, f) := \sum_{a=1}^k p_a(\boldsymbol{X},f)\hat{\mathcal{L}}(\boldsymbol{X},a)$ its estimated counterpart. Our goal is to bound its norm.

Consider the zero function $f_0=0$ and any $\boldsymbol{X}\in \mathcal{X}, a\in \mathcal{A}$ we have
$$p_a(\boldsymbol{X},f_0)= \frac{\exp(\langle \boldsymbol{w}_a, f_0 \rangle)}{\sum_{b=1}^k \exp(\langle \boldsymbol{w}_b, f_0) \rangle)}=\frac{1}{\sum_{b=1}^k 1}=\frac{1}{k}$$
which implies $\hat{\mathcal{L}}^s(\boldsymbol{X},f_0)=\sum_{a=1}^k \frac{1}{k}\hat{\mathcal{L}}(\boldsymbol{X},a)\leq 2M$
where $\hat{\mathcal{L}}(\boldsymbol{X},a) \leq 2M$ by \cref{asm: bounds-bounds} and the definition of $\mathcal{L}$.
Since $\hat{\mathcal{L}}(\boldsymbol{X},a) \geq -2M$ by the same assumption, we also have $\hat{\mathcal{L}}^s(\boldsymbol{X},f) \geq -2M$ for all $f$.
By minimality of $\hat{f}_s$ and the lower bound on the surrogate loss we have
\begin{align*}
 -2M + \lambda\Vert \hat{f}_s \Vert^2 &\leq \hat{\mathcal{R}}^s(\hat{f}_s) \leq \hat{\mathcal{R}}^s (f_0)=\frac{1}{\vert I_2 \vert} \sum_{i\in I_2} \hat{\mathcal{L}}^s(\boldsymbol{X}_i, f_0)+\lambda \underbrace{\Vert f_0 \Vert^2}_{=0} \leq 2M,
\end{align*}
hence $\lambda\Vert \hat{f}_s \Vert^2 \leq 4M$.
This gives us a bound for $\Vert \hat{f}_s\Vert \leq \sqrt{\frac{4M}{\lambda}}$.
Define
$$
\mathcal{B}^*:=\left\{f \in \mathcal{F} \colon \Vert f \Vert \leq 2\cdot \sqrt{\frac{M}{\lambda}}\right\},
$$
and observe that $\hat{f}_s \in \mathcal{B}^*$. 

To start with the decomposition we need a term that is in ``between'' $\mathcal{R}^s(\hat{f}_s)$ and $\inf_{f\in \mathcal{F}} \mathcal{R}^s(f)$. Using \cref{asm: minimizer} this motivates the definition of
\begin{align*}
f^*_\lambda := \argmin_{f\in \mathcal{F}}\mathbb{E}\left[ \mathcal{L}^s(\boldsymbol{X}, f)+\lambda \Vert f \Vert^2 \right],
\end{align*}
as the  minimizer of the theoretical surrogate risk with the regularization term $\lambda$.

We also define the regularization approximation error
\begin{align*}
a_{\mathcal{F}}(\lambda) \coloneqq \inf_{f\in\mathcal{F}} \{ \mathcal{R}^s(f)+\lambda\|f\|^2\} - \inf_{f\in\mathcal{F}}\mathcal{R}^s(f).
\end{align*}
Since $f^*_\lambda$ attains the infimum of the regularized objective we have,
\begin{align*}
a_{\mathcal{F}}(\lambda) = \mathcal{R}^s(f_\lambda^*)+\lambda\|f_\lambda^*\|^2 - \inf_{f\in\mathcal{F}}\mathcal{R}^s(f).
\end{align*}

The objective can be decomposed by introducing $f^*_\lambda$:

\begin{align}
    &\quad \mathcal{R}^s(\hat{f}_s)-\inf_{f\in \mathcal{F}} \mathcal{R}^s(f) \label{eq: loss-diff}\\
    &=\mathbb{E}[\mathcal{L}^s(\boldsymbol{X},\hat{f}_s)] - \inf_{f\in \mathcal{F}} \mathcal{R}^s(f) \notag\\
    &=\underbrace{\mathbb{E}[\mathcal{L}^s(\boldsymbol{X},\hat{f}_s)] - \mathbb{E}[\mathcal{L}^s(\boldsymbol{X},f^*_\lambda)] - \lambda \Vert f^*_\lambda \Vert^2}_{(I)}+\underbrace{\mathbb{E}[\mathcal{L}^s(\boldsymbol{X},f^*_\lambda)] + \lambda \Vert f^*_\lambda \Vert^2-\inf_{f\in \mathcal{F}} \mathcal{R}^s(f)}_{a_{\mathcal{F}}(\lambda)} \notag
\end{align}

By \cref{asm: minimizer}, $f^*_\lambda \in \mathcal{F}$ is bounded with finite norm.
We bound $(I)$ by inflating it with terms that cancel out.
\begin{align}
    (I)&=\mathbb{E}[\mathcal{L}^s(\boldsymbol{X},\hat{f}_s)] - \mathbb{E}[\mathcal{L}^s(\boldsymbol{X},f^*_\lambda)] - \lambda \Vert f^*_\lambda \Vert^2 \nonumber \\
    &\leq \mathbb{E}[\mathcal{L}^s(\boldsymbol{X},\hat{f}_s)] +\lambda\Vert \hat{f}_s\Vert^2- \mathbb{E}[\mathcal{L}^s(\boldsymbol{X},f^*_\lambda)] - \lambda \Vert f^*_\lambda \Vert^2 \nonumber \\
    &= \mathbb{E}\left[ \frac{1}{\vert I_2\vert}\sum_{i\in I_2}\mathcal{L}^s(\boldsymbol{X}_i, f^*_\lambda) - \mathbb{E}[\mathcal{L}^s(\boldsymbol{X},f^*_\lambda)]\right] \label{eq:Q1}\\
    &+\mathbb{E}\left[ \frac{1}{\vert I_2\vert}\sum_{i\in I_2}\hat{\mathcal{L}}^s(\boldsymbol{X}_i, f^*_\lambda) - \frac{1}{\vert I_2\vert}\sum_{i\in I_2}\mathcal{L}^s(\boldsymbol{X}_i, f^*_\lambda)\right] \label{eq:Q2}\\
    &+\mathbb{E}\left[ \frac{1}{\vert I_2\vert}\sum_{i\in I_2}\hat{\mathcal{L}}^s(\boldsymbol{X}_i, \hat{f}_s) + \lambda\Vert \hat{f}_s\Vert^2- \frac{1}{\vert I_2\vert}\sum_{i\in I_2}\hat{\mathcal{L}}^s(\boldsymbol{X}_i, f^*_\lambda) - \lambda \Vert f^*_\lambda \Vert^2\right] \label{eq:Q3}\\
    &+\mathbb{E}\left[ \frac{1}{\vert I_2\vert}\sum_{i\in I_2}\mathcal{L}^s(\boldsymbol{X}_i, \hat{f}_s) - \frac{1}{\vert I_2\vert}\sum_{i\in I_2}\hat{\mathcal{L}}^s(\boldsymbol{X}_i, \hat{f}_s)\right] \label{eq:Q4}\\
    &+\mathbb{E}\left[ \mathbb{E}[\mathcal{L}^s(\boldsymbol{X},\hat{f}_s)] - \frac{1}{\vert I_2\vert}\sum_{i\in I_2}\mathcal{L}^s(\boldsymbol{X}_i, \hat{f}_s)\right] \label{eq:Q5}
\end{align}
We bound each term individually

Let us start with \eqref{eq:Q3} since $\hat{f}_s$ is the minimizer for the fixed split data for any function in $\mathcal{F}$ it follows 
$$\eqref{eq:Q3} \leq 0,$$
hence also for $f^*_\lambda$.

We now condition on $\mathcal{G}_1$. By the tower property,
\begin{align*}
    \eqref{eq:Q1}
    &=\mathbb{E}\left[\mathbb{E}\left[\frac{1}{\vert I_2\vert}\sum_{i\in I_2}\mathcal{L}^s(\boldsymbol{X}_i, f^*_\lambda)\,\middle|\, \mathcal{G}_1\right]\right] - \mathbb{E}[\mathcal{L}^s(\boldsymbol{X},f^*_\lambda)]\\
    &=\mathbb{E}\left[\frac{1}{\vert I_2\vert}\sum_{i\in I_2}\mathbb{E}\left[\mathcal{L}^s(\boldsymbol{X}_i, f^*_\lambda)\,\middle|\, \mathcal{G}_1\right]\right] - \mathbb{E}[\mathcal{L}^s(\boldsymbol{X},f^*_\lambda)]\\
    &=\mathbb{E}\left[\frac{1}{\vert I_2\vert}\sum_{i\in I_2}\mathbb{E}\left[\mathcal{L}^s(\boldsymbol{X}, f^*_\lambda)\right]\right] - \mathbb{E}[\mathcal{L}^s(\boldsymbol{X},f^*_\lambda)] = 0
\end{align*}
where we used that $f^*_\lambda$ does not depend on the sample and that, for every $i\in I_2$, $\boldsymbol{X}_i$ is independent of $\mathcal{G}_1$. Notice that this argument is only possible because the function does not depend on the samples; this is different to \eqref{eq:Q5} where we have to use empirical process theory to bound the term.

For any $f \in \mathcal{F}$ we have 
\begin{align*}
    &\quad \mathbb{E}\left[ \frac{1}{\vert I_2\vert}\sum_{i\in I_2}\mathcal{L}^s(\boldsymbol{X}_i, f) - \frac{1}{\vert I_2\vert}\sum_{i\in I_2}\hat{\mathcal{L}}^s(\boldsymbol{X}_i, f)\right]\\
    &= \mathbb{E}\left[\mathbb{E}\left[\frac{1}{\vert I_2\vert}\sum_{i\in I_2}\left(\mathcal{L}^s(\boldsymbol{X}_i, f) - \hat{\mathcal{L}}^s(\boldsymbol{X}_i, f)\right)\,\middle|\, \mathcal{G}_1\right]\right]\\
    &\leq \mathbb{E}\left[\frac{1}{\vert I_2\vert}\sum_{i\in I_2}\mathbb{E}\left[\left|\mathcal{L}^s(\boldsymbol{X}_i, f) - \hat{\mathcal{L}}^s(\boldsymbol{X}_i, f)\right|\,\middle|\, \mathcal{G}_1\right]\right]\\
    &= \mathbb{E}\left[\mathbb{E}\left[\left|\mathcal{L}^s(\boldsymbol{X}, f) - \hat{\mathcal{L}}^s(\boldsymbol{X}, f)\right|\,\middle|\, \mathcal{G}_1\right]\right] \\
    &= \mathbb{E}\left[\left|\mathcal{L}^s(\boldsymbol{X}, f) - \hat{\mathcal{L}}^s(\boldsymbol{X}, f)\right|\right] \\
    &\leq O(n^{-\alpha}+n^{-\beta}),
\end{align*}
where the last inequality holds due to \cref{lemma: bound-risk-diff}.
Since the inequality holds for any $f \in \mathcal{F}$ and in particular for $f^*_\lambda, \hat{f}_s$ we can bound 
\begin{align*}
    \eqref{eq:Q2}\leq O(n^{-\alpha}+n^{-\beta})
\end{align*}
as well as 
\begin{align*}
    \eqref{eq:Q4}\leq O(n^{-\alpha}+n^{-\beta})
\end{align*}
By \cref{lemma: bound-empirical} we have 
\[
\eqref{eq:Q5} \leq O\left(\frac{1}{\sqrt{n \lambda}}\right) 
\]
Finally, combining the above results we can bound \eqref{eq: loss-diff} and the theorem follows:

\begin{align*}
    \eqref{eq: loss-diff} &= (I)+ a_\mathcal{F}(\lambda)\\
    &\leq \eqref{eq:Q1} +\eqref{eq:Q2}+\eqref{eq:Q3}+\eqref{eq:Q4}+\eqref{eq:Q5}+a_\mathcal{F}(\lambda)\\
    &\leq 0 + 0 +  O(n^{-\alpha}+n^{-\beta}) +  O(n^{-\alpha}+n^{-\beta}) + O\left(\frac{1}{\sqrt{n \lambda}}\right)  + a_\mathcal{F}
\end{align*}
\end{proof}

\begin{proof}[Proof of \cref{lemma: bound-risk-diff}]
We start by decomposing the loss
\begin{align*}
    \vert \mathcal{L}^s(\boldsymbol{X},f)-\hat{\mathcal{L}}^s(\boldsymbol{X},f)\vert
    &\leq \sum_{a=1}^k |\underbrace{p_a(\boldsymbol{X},f)}_{\leq 1}| \, \vert \mathcal{L}(\boldsymbol{X},a) - \hat{\mathcal{L}}(\boldsymbol{X},a)\vert\\
    &\leq \sum_{a=1}^k \vert \mathcal{L}(\boldsymbol{X},a) - \hat{\mathcal{L}}(\boldsymbol{X},a)\vert 
\end{align*}
Writing out the summand explicitly:
\begin{align}
    &\quad \vert \mathcal{L}(\boldsymbol{X},a)-\hat{\mathcal{L}}(\boldsymbol{X},a) \vert \nonumber\\
    &= \bigg|\max_{b \in \mathcal{A}\setminus \{a\}} U_b(\boldsymbol{X}) - L_a(\boldsymbol{X}) -  \left(  \max_{b' \in \mathcal{A}\setminus \{a\}} \hat{U}_{b'}(\boldsymbol{X}) - \hat{L}_a(\boldsymbol{X}) \right) \bigg| \nonumber\\
    &\leq \bigg|\max_{b \in \mathcal{A}\setminus \{a\}} U_b(\boldsymbol{X}) - \max_{b' \in \mathcal{A}\setminus \{a\}} \hat{U}_{b'}(\boldsymbol{X}) \bigg| + \bigg|  \hat{L}_a(\boldsymbol{X}) -L_a(\boldsymbol{X}) \bigg| \label{eq: risk-dif}
\end{align}
Let $b=\arg \max_{b \in \mathcal{A}\setminus \{a\}} U_b(\boldsymbol{X})$ and $b'=\arg \max_{b \in \mathcal{A}\setminus \{a\}} \hat{U}_b(\boldsymbol{X})$, then there are two cases either $b=b'$ for which we have:
\begin{align*}
    \eqref{eq: risk-dif} = \vert U_b(\boldsymbol{X}) - \hat{U}_{b}(\boldsymbol{X}) \vert + \vert  \hat{L}_a(\boldsymbol{X}) -L_a(\boldsymbol{X}) \vert
\end{align*}
or $b\neq b'$ for which we have that either $U_b(\boldsymbol{X})\leq \hat{U}_{b'}(\boldsymbol{X})$  hence
\begin{align*}
    \eqref{eq: risk-dif} &= \vert U_b(\boldsymbol{X}) - \hat{U}_{b'}(\boldsymbol{X}) \vert + \vert  \hat{L}_a(\boldsymbol{X}) -L_a(\boldsymbol{X}) \vert \\
    &\leq \vert U_{b'}(\boldsymbol{X}) - \hat{U}_{b'}(\boldsymbol{X}) \vert + \vert  \hat{L}_a(\boldsymbol{X}) -L_a(\boldsymbol{X}) \vert
\end{align*}
or we have $\hat{U}_{b'}(\boldsymbol{X}) < U_b(\boldsymbol{X}) $  hence
\begin{align*}
    \eqref{eq: risk-dif} &= \vert U_b(\boldsymbol{X}) - \hat{U}_{b'}(\boldsymbol{X}) \vert + \vert  \hat{L}_a(\boldsymbol{X}) -L_a(\boldsymbol{X}) \vert \\
    &\leq \vert U_{b}(\boldsymbol{X}) - \hat{U}_{b}(\boldsymbol{X}) \vert + \vert  \hat{L}_a(\boldsymbol{X}) -L_a(\boldsymbol{X}) \vert
\end{align*}
Hence in either cases \eqref{eq: risk-dif} can be bounded in expectation using \cref{asm: convergence} and the convergence rate follows. 
\end{proof}

\begin{proof}[Proof of \cref{lemma: lipschitz-covering}]
    Let $N_{c} =  N\left(\frac{\epsilon}{L_{lip}} , \mathcal{B}^*, \Vert \cdot \Vert_{\mathcal{B}^*}\right)$. There must exist a covering $f_1, \dots, f_{N_{c}}$ such that for any $f \in \mathcal{B}^*$  there exits a $f_i$ such that $\Vert f -f_i \Vert_{\mathcal{B}^*} \leq \frac{\epsilon}{L_{lip}}$.
    For any $g \in \mathcal{F}_{\mathcal{L}}$ there must exist a $f$ such that $g=\Phi(f)$ for some $f \in \mathcal{B}^*$. If we choose the closest $f_i$ to $f$ in $\Vert \cdot \Vert_{\mathcal{B}^*}$ we have
    \begin{align*}
        \Vert g-\Phi(f_i) \Vert_{\mathcal{F}_{\mathcal{L}}}= \Vert \Phi(f)- \Phi(f_i)\Vert_{\mathcal{F}_{\mathcal{L}}}\leq L_{lip} \Vert f- f_i\Vert_{\mathcal{B}^*} \leq \epsilon,
    \end{align*}
    by the assumption of the covering of $\mathcal{B}^*$. Hence $\Phi(f_i)$ is an $\epsilon$-covering and the lemma follows.
\end{proof}

\begin{proof}[Proof of \cref{lemma: bound-empirical}]
Define the set of functions motivated by the proof of \cref{thm:oracle-inequality}
$$
\mathcal{B}^*:=\left\{f \in \mathcal{F} \colon \Vert f \Vert \leq 2\cdot \sqrt{\frac{M}{\lambda}}\right\}.
$$
In particular notice that $\hat{f}_s \in \mathcal{B}^*$ by construction.
Then
\begin{align}
     \mathbb{E}\left[ \mathbb{E}[\mathcal{L}^s(\boldsymbol{X},\hat{f}_s)] - \frac{1}{\vert I_2\vert}\sum_{i\in I_2}\mathcal{L}^s(\boldsymbol{X}_i, \hat{f}_s)\right] &\leq \mathbb{E}\left[ \sup_{f \in \mathcal{B}^*} \bigg|\mathbb{E}[\mathcal{L}^s(\boldsymbol{X},f)] - \frac{1}{\vert I_2\vert}\sum_{i\in I_2}\mathcal{L}^s(\boldsymbol{X}_i, f) \bigg| \right] \label{eq:sup-bound}
\end{align}
Consider the function class $\mathcal{F}_{\mathcal{L}} := \{\mathcal{L}^s(\cdot, f) \mid  f\in \mathcal{B}^*\}$.
By defining the empirical process as 
\begin{align*}
\|\mathbb{G}_{I_2}\|_{\mathcal{F}_{\mathcal{L}}}
&:= \sup_{g \in \mathcal{F}_{\mathcal{L}}} 
\left| \mathbb{G}_{I_2} g \right| \\
&= \sup_{f \in \mathcal{B}^*} 
\left| \mathbb{G}_{I_2}\bigl(\mathcal{L}^s(\cdot, f)\bigr) \right|.
\end{align*}
Then
\begin{align*}
\sup_{f \in \mathcal{B}^*} 
\left| P\mathcal{L}^s(\cdot, f) 
- \mathbb{P}_{I_2}\mathcal{L}^s(\cdot, f) \right|
=
\frac{1}{\sqrt{|I_2|}}
\|\mathbb{G}_{I_2}\|_{\mathcal{F}_{\mathcal{L}}}.
\end{align*}

and we can write \cref{eq:sup-bound} in terms of the empirical process,

\begin{align*}
    \mathbb{E}\left[ \sup_{f \in \mathcal{B}^*} \bigg|\mathbb{E}[\mathcal{L}^s(\boldsymbol{X},f)] - \frac{1}{\vert I_2\vert}\sum_{i\in I_2}\mathcal{L}^s(\boldsymbol{X}_i, f) \bigg| \right]
    =  \mathbb{E}\left[ \frac{1}{\sqrt{|I_2|}} \Vert\mathbb{G}_{I_2} \Vert_{\mathcal{F}_{\mathcal{L}}} \right]
\end{align*}

To apply theorem 2.14.1 of \cite{vanderVaart1996} we need an envelop function for $\vert \mathcal{L}^s(\boldsymbol{X}, f) \vert$ for all $\mathcal{L}^s(\cdot, f) \in \mathcal{F}_{\mathcal{L}}$. The function $F(\boldsymbol{X}) := 2M $ is a valid envelop since,
\begin{align*}
    \vert \mathcal{L}^s(\boldsymbol{X},f)\vert
    &\leq \sum_{a=1}^k \vert p_a(\boldsymbol{X},f) \mathcal{L}(\boldsymbol{X},a)\vert\\
    &\leq \sum_{a=1}^k \vert  p_a(\boldsymbol{X},f) \cdot 2M\vert \\
    &=2M
\end{align*}

Applying Theorem 2.14.1 of \cite{vanderVaart1996}  we have 
\begin{align}
\label{eq: entropy-integral}
\mathbb{E} \Vert\mathbb{G}_{I_2} \Vert_{\mathcal{F}_{\mathcal{L}}}\leq O\left( \sup_Q \int_{0}^{1} \sqrt{1+\log N(\epsilon \cdot \|F\|_{Q,2}, \mathcal{F}_{\mathcal{L}}, L_2(Q))} \,d\epsilon \Vert F \Vert_{P,2} \right),
\end{align}
where $\|f\|_{Q,2} := \left( \int |f|^2 \, dQ \right)^{1/2}$ and $\|f\|_{P,2} := \left( \int |f|^2 \, dP \right)^{1/2}$ denote the $L_2(Q)$ and $L_2(P)$ norms, respectively.
Notice that the envelop function satisfies $\Vert F \Vert_{P,2} = \sqrt{\mathbb{E}[F(\boldsymbol{X})^2]}=2M$.

To calculate the integral we want to use \cref{lemma: lipschitz-covering}. For this, fix $\boldsymbol X$ and define
\begin{align*}
\psi_{\boldsymbol X}(v) \coloneqq \sum_{a=1}^k p_a(\boldsymbol X,v)\mathcal{L}(\boldsymbol X,a), \qquad v\in\mathbb R^{k-1},
\end{align*}
so that
\begin{align*}
\mathcal{L}^s(\boldsymbol{X},f)=\psi_{ \boldsymbol{X} }(f(\boldsymbol{X})).
\end{align*}
Differentiating with respect to $v$ yields
\begin{align*}
\nabla_v \psi_{\boldsymbol X}(v) = \sum_{a=1}^k \mathcal{L}(\boldsymbol X,a)\,p_a(\boldsymbol X,v) \left( \boldsymbol w_a-\sum_{b=1}^k p_b(\boldsymbol X,v)\boldsymbol w_b \right).
\end{align*}
Using \cref{asm: bounds-bounds} and $\|\boldsymbol w_a\|_2=1$, we obtain
\begin{align*}
\|\nabla_v \psi_{\boldsymbol X}(v)\|_2 
&\leq \sum_{a=1}^k 2M\,p_a(\boldsymbol X,v) \left( \|\boldsymbol w_a\|_2+ \left\|\sum_{b=1}^k p_b(\boldsymbol X,v)\boldsymbol w_b\right\|_2 \right) \\
&\leq \sum_{a=1}^k 2M\,p_a(\boldsymbol X,v)(1+1) =4M.
\end{align*}
Hence $\psi_{\boldsymbol X}$ is $4M$-Lipschitz on $\mathbb R^{k-1}$. Therefore, for any $f_1,f_2$,
\begin{align*}
|\mathcal{L}^s(\boldsymbol X,f_2)-\mathcal{L}^s(\boldsymbol X,f_1)|
&= |\psi_{\boldsymbol X}(f_2(\boldsymbol X))-\psi_{\boldsymbol X}(f_1(\boldsymbol X))| \\
&\leq 4M\,\|f_2(\boldsymbol X)-f_1(\boldsymbol X)\|_2.
\end{align*}
Squaring and integrating with respect to $Q$ gives
\begin{align*}
\|\mathcal{L}^s(\cdot,f_2)-\mathcal{L}^s(\cdot,f_1)\| _{L_2(Q)}
\leq 4M\,\|f_2-f_1\|_{L_2(Q)}.
\end{align*}

This shows that the map $f \mapsto \mathcal{L}^s(\cdot,f)$ from $\mathcal{B}^* \to \mathcal{F}_{\mathcal{L}}$ is Lipschitz with constant $L_{lip}=4M$ with respect to the $L_2(Q)$ norm.
This allows us to use \cref{lemma: lipschitz-covering} to bound
\begin{align*}
N(\epsilon \cdot \|F\|_{Q,2}, \mathcal{F}_{\mathcal{L}}, L_2(Q)) \leq N\left(\frac{\epsilon}{L_{lip}} \cdot \|F\|_{Q,2}, \mathcal{B}^*, L_2(Q)\right)
\end{align*}

We can bound this covering number further using $\mathcal{B}_0$ by plugging in the the definition of $\mathcal{B}^*$. 

\begin{align*}
\log N\left(\frac{\epsilon}{L_{lip}}\|F\|_{Q,2},\mathcal{B}^*,L_2(Q)\right)
&=\log N\left(\frac{\epsilon \|F\|_{Q,2}\sqrt{\lambda}}{2L_{lip}\sqrt{M}},\mathcal{B}_0,L_2(Q)\right) \\
&\leq O\left(\left(\frac{\epsilon \|F\|_{Q,2}\sqrt{\lambda}}{2L_{lip}\sqrt{M}}\right)^{-v} \right) 
=O\left(\lambda^{-v/2}\epsilon^{-v}\right).
\end{align*}

Where we have used \cref{asm: entropy} for the last inequality and the fact that in general $N ( \varepsilon, r\cdot A, d) = N(\frac{\varepsilon}{r}, A,d)$.

The integral becomes \begin{align*}
    \int_0^1 \sqrt{\lambda^{-\frac{v}{2}} \epsilon ^{-v} }d\epsilon= \lambda^{-\frac{v}{4}} \int_0^1\epsilon^{-\frac{v}{2}}d\epsilon = O(\lambda^{-\frac{v}{4}} ) \leq O(\lambda^{-\frac{1}{2}}),
\end{align*}
where we have used the fact that the integral converges for $-v/2<1$ which is satisfied by the assumption $0<v<2$. The last step is only justifiable if $ 0 < \lambda \leq 1$. 
We can bound
\begin{align*}
    \eqref{eq: entropy-integral} \leq O(\lambda^{-\frac{1}{2}} \cdot \Vert F\Vert_{P,2})= O(\lambda^{-\frac{1}{2}} \cdot 2M)= O(\lambda^{-\frac{1}{2}})
\end{align*}
Putting everything together we get
\begin{align*}
    \eqref{eq:Q5} &\leq \mathbb{E}\left[ \frac{1}{\sqrt{|I_2|}} \Vert\mathbb{G}_{I_2} \Vert_{\mathcal{F}_{\mathcal{L}}} \right]\\
    &\leq O\left(\frac{1}{\sqrt{|I_2|}}\cdot \lambda^{-\frac{1}{2}}\right)\\
    &=O\left(\frac{1}{\sqrt{n \lambda}}\right),
\end{align*}
where the last step uses $|I_2|= O(n)$.
\end{proof}

\begin{proof} [Proof of the claim in \cref{ex:gaussian-verification}]
We show each assumption separately.
\paragraph{\Cref{asm: entropy}: }
First we show condition \cref{asm: entropy} is met  for the product hypothesis space $\mathcal{F} = \mathcal{H}^{k-1}$, where $\mathcal{H}$ is the RKHS induced by the Gaussian RBF kernel $k(\boldsymbol{X}, \boldsymbol{t}) = \exp(-\gamma \|\boldsymbol{X} - \boldsymbol{t}\|^2)$ on a compact domain $\mathcal{X} \subset \mathbb{R}^d$.

The proof proceeds in two steps. First, we establish the entropy bound for the unit ball of the scalar component $\mathcal{H}$ using the VC-hull theory from \cite{Kosorok2008}. Second, we extend this result to the product space $\mathcal{F}$.

Let $\mathcal{B}_{\mathcal{H}} = \{h \in \mathcal{H} : \|h\|_{\mathcal{H}} \leq 1\}$ be the unit ball of the scalar RKHS. We identify $\mathcal{B}_{\mathcal{H}}$ as a subset of a VC-hull class. 

Define the VC-Subgraph class $\mathcal{G}$  associated with $\mathcal{B}_\mathcal{H}$ as the set of kernel functions indexed by the covariates:
\begin{align*}
    \mathcal{G} = \{ \phi_{\boldsymbol{t}}(\cdot) : \mathcal{X} \to \mathbb{R} \mid \phi_{\boldsymbol{t}}(\boldsymbol{X}) = \exp(-\gamma \Vert \boldsymbol{X} - \boldsymbol{t}\Vert^2), \, \boldsymbol{t} \in \mathbb{R}^d \}.
\end{align*}
We first show that $\mathcal{G}$ is a VC-subgraph class with finite VC-index $V$. Expanding the exponent, we have:
\begin{align*}
    \phi_{\boldsymbol{t}}(\boldsymbol{X}) = \exp\left( -\gamma (\Vert\boldsymbol{X}\Vert^2 - 2\boldsymbol{X}^\top \boldsymbol{t} + \Vert \boldsymbol{t}\Vert^2) \right).
\end{align*}
The term inside the exponential, $h_{\boldsymbol{t}}(\boldsymbol{X}) = -\Vert\boldsymbol{X}\Vert^2 + 2\boldsymbol{X}^\top \boldsymbol{t} - \Vert\boldsymbol{t}\Vert^2$, belongs to the finite-dimensional vector space of functions spanned by $\{1, X_1, \dots, X_d, \|\boldsymbol{X}\|^2\}$. This space has dimension $d+2$. By Lemma 9.6 of Kosorok \cite{Kosorok2008}, a vector space of functions of dimension $k$ has VC-index $V \leq k+2$. Thus, the family of functions inside the exponent has finite VC-index. Furthermore, since $u \mapsto \exp(\gamma u)$ is a monotone transformation, Lemma 9.9 of Kosorok \cite{Kosorok2008} implies that the class $\mathcal{G}$ preserves the VC-subgraph property with a finite index $V < \infty$.

The unit ball $\mathcal{B}_{\mathcal{H}}$ is contained in the closed convex hull of $\mathcal{G}$. Therefore, $\mathcal{B}_{\mathcal{H}}$ is a VC-hull class. 

We apply Corollary 9.5 of \cite{Kosorok2008}, which states that for a VC-hull class generated by a base class with VC-index $V$, the covering number satisfies:
\begin{align*}
    \sup_Q \log N(\epsilon \Vert F\Vert_{Q,2}, \mathcal{B}_{\mathcal{H}}, L_2(Q)) \leq K \left(\frac{1}{\epsilon}\right)^{2 - \frac{2}{V}},
\end{align*}
where $K < \infty$ is a constant depending only on $V$. Since $V < \infty$, the exponent $ 2 - 2/V$ is strictly less than 2. We choose as an envelop function $F=1$, since $|f(\boldsymbol{X})|  = |\langle f, k(\cdot, \boldsymbol{X})| \leq \Vert f \Vert_\mathcal{H} \sqrt{k(\boldsymbol{X}, \boldsymbol{X})} \leq 1$

We now extend this result to the product space $\mathcal{F} = \mathcal{H}^{k-1}$. The norm on $\mathcal{F}$ is $\|f\|_{\mathcal{F}}^2 = \sum_{j=1}^{k-1} \|f^j\|_{\mathcal{H}}^2$.
Consider the unit ball $\mathcal{B}_{\mathcal{F}} = \{f \in \mathcal{F} : \|f\|_{\mathcal{F}} \leq 1\}$. If $f \in \mathcal{B}_{\mathcal{F}}$, then $\|f^j\|_{\mathcal{H}}^2 \leq 1$ for all $j=1, \dots, k-1$. Therefore, the unit ball of the product space is contained in the cartesian product of the scalar unit balls:
\begin{align*}
    \mathcal{B}_{\mathcal{F}} \subseteq \underbrace{\mathcal{B}_{\mathcal{H}} \times \dots \times \mathcal{B}_{\mathcal{H}}}_{k-1 \text{ times}}.
\end{align*}

Let $\delta = \epsilon / \sqrt{k-1}$ and $\mathcal{C}$ be a minimal $\delta$-cover of the scalar unit ball $\mathcal{B}_{\mathcal{H}}$ with respect to the scalar $L_2(Q)$ norm. For any $f \in \mathcal{B}_{\mathcal{F}}$, each component $f^j$ belongs to $\mathcal{B}_{\mathcal{H}}$, so there exists a $g^j \in \mathcal{C}$ such that $\|f^j - g^j\|_{L_2(Q)} \leq \delta$. Let $g = (g^1, \dots, g^{k-1}) \in \mathcal{C}_{\text{prod}}$. The distance in the product space is:
\begin{align*}
    \|f - g\|_{L_2(Q)} = \sqrt{\sum_{j=1}^{k-1} \|f^j - g^j\|_{L_2(Q)}^2} \leq \sqrt{\sum_{j=1}^{k-1} \delta^2} = \sqrt{k-1} \, \delta = \epsilon.
\end{align*}
Thus, $\mathcal{C}_{\text{prod}}$ is an $\epsilon$-cover for $\mathcal{B}_{\mathcal{F}}$. Hence we have $N(\epsilon, \mathcal{B}_{\mathcal{F}}, \|\cdot\|_{L_2(Q)}) \leq N\left(\frac{\epsilon}{\sqrt{k-1}}, \mathcal{B}_{\mathcal{H}}, \|\cdot\|_{L_2(Q)}\right)^{k-1}$
Substituting the scalar bound from Step 1, $\sup_Q \log N(\delta, \mathcal{B}_{\mathcal{H}}) \leq A \delta^{-v}$, yields:
\begin{align}
    \sup_Q \log N(\epsilon, \mathcal{B}_{\mathcal{F}}, \|\cdot\|_{L_2(Q)}) 
    &\leq (k-1) A \left(\frac{\epsilon}{\sqrt{k-1}}\right)^{-v} = O(\epsilon^{-v}).
\end{align}

\paragraph{\Cref{asm:emp-minimizer}: }
We note that \cref{asm:emp-minimizer} is exactly met by the Riesz representer theorem 

\Cref{sec:kernelized-formulation}.
\paragraph{\Cref{asm: norm} :}
We show that \cref{asm: norm} holds. Observe that $|f^j(\boldsymbol{X})| = |\langle f^j, k(\boldsymbol{X}, \cdot) \rangle_{\mathcal{H}}| \leq \|f^j\|_{\mathcal{H}} \sqrt{k(\boldsymbol{X}, \boldsymbol{X})} = \|f^j\|_{\mathcal{H}}$.

It follows from the definition of the norm that for any $j$, $\|f^j\|_{\mathcal{H}} \leq \|f\|_{\mathcal{F}}$.
Therefore, $\|f\|_\infty = \sup_{\boldsymbol{X}} \max_{j} |f^j(\boldsymbol{X})| \leq \max_j \|f^j\|_{\mathcal{H}} \leq \|f\|_{\mathcal{F}}$, satisfying the assumption with $\kappa=1$.

\paragraph{\Cref{asm: minimizer}}Last we show that \cref{asm: minimizer} holds.
We have to show that the minimizer exists. We define
\begin{align}
\label{eq:minimizer-vehicle}
J_\lambda(f):=\mathcal{R}^s(f)+\lambda\|f\|_{\mathcal{F}}^2.
\end{align}
Let $(f_n)$ be a minimizing sequence for $J_\lambda$. Using \cref{asm: bounds-bounds},
\begin{align*}
\mathcal{R}^s(f)\geq -2M.
\end{align*}
and therfor 
\begin{align*}
J_\lambda(f)\geq -2M+\lambda\|f\|_{\mathcal{F}}^2.
\end{align*}
As $(f_n)$ is minimizing and $J_\lambda(f_n)$ is bounded we have $(f_n)$ is bounded in $\mathcal{F}$.

Since $\mathcal{F}$ is a Hilbert space, so every bounded sequence admits a weakly convergent subsequence. Thus $\exists f\in\mathcal{F}$ such that
\begin{align*}
f_n \rightharpoonup f \qquad \text{weakly in }\mathcal{F}.
\end{align*}
This also implies pointwise convergence. To  see this, write $f_n=(f_n^1,\dots,f_n^{k-1})$ and $f=(f^1,\dots,f^{k-1})$, for every fixed $\boldsymbol X\in\mathcal{X}$ and every component $j$,
\begin{align*}
f_n^j(\boldsymbol X) =\langle f_n^j, k(\boldsymbol X,\cdot)\rangle_{\mathcal{H}} \to \langle f^j, k(\boldsymbol X,\cdot)\rangle_{\mathcal{H}}= f^j(\boldsymbol X),
\end{align*}
By the reproducing property. Thus
\begin{align*}
f_n(\boldsymbol X)\to f(\boldsymbol X)
\qquad\text{for every }\boldsymbol X.
\end{align*}
By continuity of the softmax map, this yields
\begin{align*}
\mathcal{L}^s(\boldsymbol X,f_n)\to \mathcal{L}^s(\boldsymbol X,f) \qquad\text{for every } \boldsymbol X.
\end{align*}
Using \cref{asm: bounds-bounds} which can be applied to $|\mathcal{L}^s(\boldsymbol X,f_n)|\leq 2M$ 
we can use dominated convergence theorem so
\begin{align*}
\mathcal{R}^s(f_n)\to \mathcal{R}^s(f).
\end{align*}
For the other term in \cref{eq:minimizer-vehicle} weak convergence implies
\begin{align*}
\langle f_n,f\rangle_{\mathcal{F}}\to \langle f,f\rangle_{\mathcal{F}}=\|f\|_{\mathcal{F}}^2.
\end{align*}
By Cauchy--Schwarz,
\begin{align*}
\langle f_n,f\rangle_{\mathcal{F}}\leq \|f_n\|_{\mathcal{F}}\|f\|_{\mathcal{F}}.
\end{align*}
Taking the limit gives
\begin{align*}
\|f\|_{\mathcal{F}}^2 \leq \left(\inf_{n\to\infty}\|f_n\|_{\mathcal{F}}\right)\|f\|_{\mathcal{F}},
\end{align*}
and thus
\begin{align*}
\|f\|_{\mathcal{F}}\leq \inf{n\to\infty}\|f_n\|_{\mathcal{F}}.
\end{align*}

Combining allows us to show that $f$ attains the minimum:
\begin{align*}
J_\lambda(f) &= \mathcal{R}^s(f)+\lambda\|f\|_{\mathcal{F}}^2 \\
&\leq \inf_{n\to\infty} (
\mathcal{R}^s(f_n)+\lambda\|f_n\|_{\mathcal{F}}^2 ).
\end{align*}
    
\end{proof}

\section{Simulation Setup Details}

\label{apx:simulation-setup}
\begin{figure}[h]
    \centering
    \begin{tikzpicture}[>=stealth, node distance=2cm, every node/.style={font=\small}]
        \node (Z) {$Z$};
        \node (A) [right of=Z] {$A$};
        \node (Y) [right of=A] {$Y$};
        \node (X) [above of=A, yshift=-0.5cm] {$\boldsymbol{X}$};
        \node (U) [below of=A, yshift=0.5cm] {$U$};
    
        \draw[->] (Z) -- (A);
        \draw[->] (A) -- (Y);
        \draw[->] (X) -- (A);
        \draw[->] (X) -- (Y);
        \draw[->, dashed] (U) -- (A);
        \draw[->, dashed] (U) -- (Y);
    \end{tikzpicture}
    
    \caption{DAG of the data generating process used in the simulation.}
    \label{fig:simulation-DAG}
\end{figure}
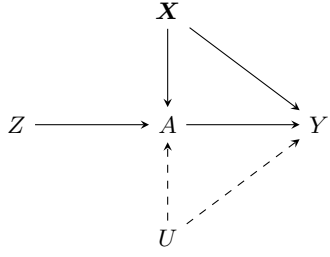

This appendix provides a detailed description of the data generating process used in the numerical experiments of \cref{sec:experiments}. 
The data generating process encodes the DAG \cref{fig:simulation-DAG}. In addition to illustration the method figure \cref{fig:latent-representation} shows the mapping of the covariates into the latent space. 

\begin{figure}[t]
    \centering
    \includegraphics[width=0.5\linewidth]{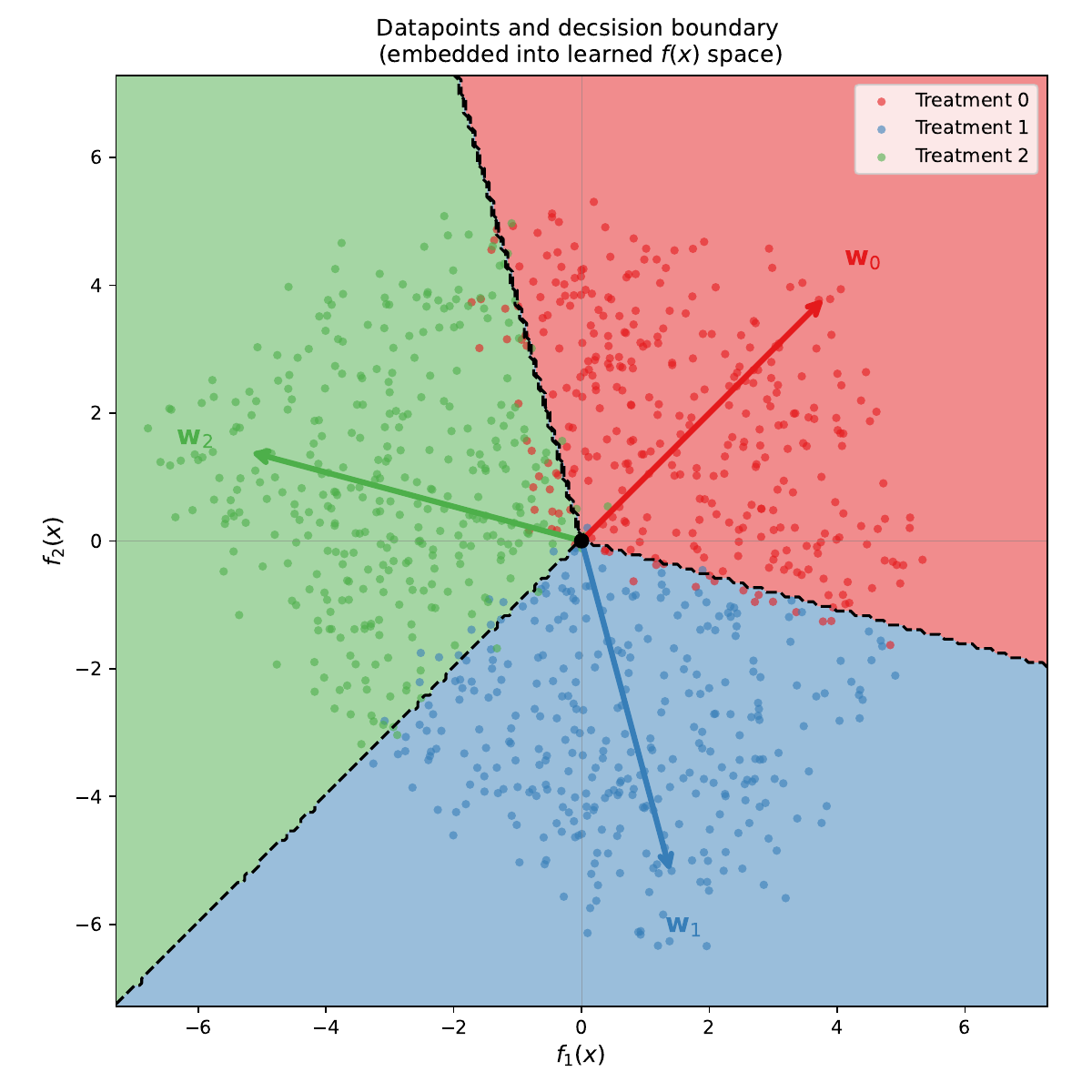}
    \caption{Illustration of the learned latent-space representation for a simulated dataset with $k=3$ treatment options. Each covariate vector $\boldsymbol{X} \in \mathbb{R}^d$ from the generated dataset is mapped by the learned function $f$ into the latent space $\mathbb R^{k-1}=\mathbb R^2$. The three treatment classes are represented by the simplex vertices $w_1,w_2,w_3$, corresponding to the embedding of the action space under $\phi$. The decision boundaries partition the latent space into three regions, where each mapped point $f(\boldsymbol{X})$ is assigned to the action whose simplex vertex yields the largest inner product.}    
    \label{fig:latent-representation}
\end{figure}

\subsection{Structural Equations}

The exogenous variables are sampled independently: $\boldsymbol{X} \sim \text{Uniform}(-1, 1)^d$ with $d = 5$, $Z \sim \text{Uniform}\{0, \ldots, |Z|-1\}$ with $|Z| = 3$, and $U \sim \text{Categorical}(\{-2, -1, 0, 1, 2\}, \text{softmax}(1, 2, 4, 2, 1))$.

\paragraph{Treatment assignment.}
The treatment $A \in \{0, \ldots, k-1\}$ is drawn from a multinomial logistic model:
\begin{align}
    P(A = a \mid \boldsymbol{X}, Z, U) = \frac{\exp(\ell_a)}{\sum_{b} \exp(\ell_b)}, \quad
    \ell_a = \boldsymbol{X}^\top \boldsymbol{\beta}_{A,a} + \boldsymbol{z}^\top \boldsymbol{\Gamma}_{\cdot, a} + \delta_a \cdot U, \label{eq:dgp-treatment}
\end{align}
where $\boldsymbol{z}$ is the one-hot encoding of $Z$, $\boldsymbol{\beta}_A \in \mathbb{R}^{d \times k}$ has entries drawn i.i.d.\ from $\text{Uniform}(-0.5, 0.5)$, and $\delta_a = \lambda_{\text{conf}} \cdot (2a - k + 1)/(k-1)$ produces linearly spaced confounding from $-\lambda_{\text{conf}}$ to $+\lambda_{\text{conf}}$. The instrument strength matrix is:
\begin{align}
    \boldsymbol{\Gamma} = \lambda_{\text{iv}} \cdot [\boldsymbol{\Gamma}_{\text{full}}]_{1:|Z|, \, 1:k}, \quad
    \boldsymbol{\Gamma}_{\text{full}} = \begin{pmatrix}
    1 & 0 & 0 & 1 & 0 & 1 & 1 \\
    0 & 1 & 0 & 1 & 1 & 0 & 1 \\
    0 & 0 & 1 & 1 & 1 & 1 & 1 \\
    0 & 0 & 0 & 0 & 0 & 0 & 0
    \end{pmatrix}, \label{eq:dgp-iv-matrix}
\end{align}
so that for $k \leq |Z|$ each treatment has a dedicated instrument level, while for $k > |Z|$ additional treatments share variation across levels.

\paragraph{Outcome.}
The binary outcome $Y$ is drawn from:
\begin{align}
    P(Y = 1 \mid \boldsymbol{X}, A = a, U) = \sigma\!\left(\boldsymbol{X}^\top \boldsymbol{\beta}_{Y,a} + \lambda_{\text{conf}} \cdot U\right), \label{eq:dgp-outcome}
\end{align}
where $\sigma$ is the logistic sigmoid. To create treatment effect heterogeneity, each treatment's coefficient vector $\boldsymbol{\beta}_{Y,a}$ loads strongly (with scale $s = 2$) on a distinct subset of $\lfloor d/k \rfloor$ consecutive covariates and has moderate negative entries elsewhere:
\begin{align}
    [\boldsymbol{\beta}_Y]_{j,a} = \begin{cases}
        s \cdot (1.5 + \epsilon_{j,a}), \quad \epsilon_{j,a} \sim \text{Uniform}(0, 0.5), & \text{if } j \text{ assigned to } a, \\
        -s \cdot \eta_{j,a}, \quad \eta_{j,a} \sim \text{Uniform}(0.3, 0.7), & \text{otherwise}.
    \end{cases} \label{eq:dgp-beta-Y}
\end{align}
This ensures clear regions where each treatment is optimal while allowing overlap.

\subsection{Ground Truth and Bounds}

Since $U$ has finite support, the true potential outcomes are computed exactly:
\begin{align}
    \mu_a(\boldsymbol{X}) = \sum_{u} P(Y = 1 \mid \boldsymbol{X}, A = a, U = u) \cdot P(U = u). \label{eq:dgp-ground-truth}
\end{align}
Similarly, the true joint probabilities $P(A = a, Y = y \mid Z = z, \boldsymbol{X})$ are obtained by marginalizing over $U$. These exact probabilities are supplied to \texttt{causaloptim} \cite{Sachs2022}, which computes sharp bounds $[L_a(\boldsymbol{X}), U_a(\boldsymbol{X})]$ by solving the linear program implied by the causal graph. For the MOML baseline, we supply the true propensity scores $P(A \mid \boldsymbol{X})$, obtained by additionally marginalizing over $Z$.

\subsection{Non-Smooth Outcome Variant}
\label{apx:non-smooth}

To test sensitivity to smoothness , we add a high-frequency perturbation to the outcome:
\begin{align}
    P(Y = 1 \mid \boldsymbol{X}, A = a, U) = \sigma\!\left(\boldsymbol{X}^\top \boldsymbol{\beta}_{Y,a} + \lambda_{\text{conf}} \cdot U + \text{sign}\!\left(\sin\!\left(5 \cdot \boldsymbol{X}^\top \boldsymbol{\beta}_{Y,a}\right)\right) \cdot 2\right). \label{eq:dgp-nonsmooth}
\end{align}
Here $\text{sign}(\sin(\cdot))$ term creates a rough square wave pattern, generating sharp transitions that smooth learners cannot capture well.

\subsection{Summary of Simulation Parameters}

\begin{table}[ht]
\centering
\caption{Simulation parameters. Each experiment varies one parameter (bold) while holding the others at their default values.}
\label{tab:dgp-parameters}
\begin{tabular}{lccccc}
\toprule
Parameter & Default & Sample Size & Confounding & IV Strength & Treatments \\
\midrule
$n$ (training size) & 12{,}000 & 5-20{,}000 & 12{,}000 & 12{,}000 & 12{,}000 \\
$\lambda_{\text{conf}}$ (confounding) & 4 & 0 & 0-10 & 4 & 4 \\
$\lambda_{\text{iv}}$ (IV strength) & 1 & 1 & 1 & 0-7 & 1 \\
$k$ (treatments) & 3 & 3 & 3 & 3 & 2-6 \\
$|Z|$ (IV levels) & 3 & 3 & 3 & 3 & 3 \\
$d$ (covariates) & 5 & 5 & 5 & 5 & 5 \\
$n_{\text{test}}$ (test size) & 5{,}000 & 5{,}000 & 5{,}000 & 5{,}000 & 5{,}000 \\
MC repeats & 10 & 10--200\textsuperscript{$\dagger$} & 10 & 10 & 10 \\
\bottomrule
\end{tabular}

\vspace{0.3em}
{\footnotesize \textsuperscript{$\dagger$}More repetitions for small $n$ to account for higher variance, when computation time is short.}
\end{table}

\printbibliography 

@inproceedings{Chen2019,
 author = {Chen, Minmin and Gummadi, Ramki and Harris, Chris and Schuurmans, Dale},
 booktitle = {Advances in Neural Information Processing Systems},
 editor = {H. Wallach and H. Larochelle and A. Beygelzimer and F. d\textquotesingle Alch\'{e}-Buc and E. Fox and R. Garnett},
 pages = {},
 publisher = {Curran Associates, Inc.},
 title = {Surrogate Objectives for Batch Policy Optimization in One-step Decision Making},
 url = {https://proceedings.neurips.cc/paper_files/paper/2019/file/84899ae725ba49884f4c85c086f1b340-Paper.pdf},
 volume = {32},
 year = {2019}
}

@article{Liu1989,
  title = {On the limited memory BFGS method for large scale optimization},
  volume = {45},
  ISSN = {1436-4646},
  url = {http://dx.doi.org/10.1007/BF01589116},
  DOI = {10.1007/bf01589116},
  number = {1–3},
  journal = {Mathematical Programming},
  publisher = {Springer Science and Business Media LLC},
  author = {Liu,  Dong C. and Nocedal,  Jorge},
  year = {1989},
  month = aug,
  pages = {503–528}
}

@BOOK{Bengio2016,
  title     = "Deep Learning",
  author    = "Bengio, Yoshua",
  publisher = "MIT Press",
  series    = "Adaptive Computation and Machine Learning series",
  month     =  nov,
  year      =  2016,
  address   = "London, England"
}

@article{Arlot2010,
  title = {A survey of cross-validation procedures for model selection},
  volume = {4},
  ISSN = {1935-7516},
  url = {http://dx.doi.org/10.1214/09-SS054},
  DOI = {10.1214/09-ss054},
  number = {none},
  journal = {Statistics Surveys},
  publisher = {Institute of Mathematical Statistics},
  author = {Arlot,  Sylvain and Celisse,  Alain},
  year = {2010},
  month = jan 
}

@article{Jonzon2025,
  title = {Adding covariates to bounds: what is the question?},
  volume = {13},
  ISSN = {2193-3685},
  url = {http://dx.doi.org/10.1515/jci-2025-0004},
  DOI = {10.1515/jci-2025-0004},
  number = {1},
  journal = {Journal of Causal Inference},
  publisher = {Walter de Gruyter GmbH},
  author = {Jonzon,  Gustav and Gabriel,  Erin E. and Sj\"{o}lander,  Arvid and Sachs,  Michael C.},
  year = {2025},
  month = jan 
}

@article{Levis2025,
  title = {Covariate-assisted bounds on causal effects with instrumental variables},
  volume = {87},
  ISSN = {1467-9868},
  url = {http://dx.doi.org/10.1093/jrsssb/qkaf028},
  DOI = {10.1093/jrsssb/qkaf028},
  number = {5},
  journal = {Journal of the Royal Statistical Society Series B: Statistical Methodology},
  publisher = {Oxford University Press (OUP)},
  author = {Levis,  Alexander W and Bonvini,  Matteo and Zeng,  Zhenghao and Keele,  Luke and Kennedy,  Edward H},
  year = {2025},
  month = may,
  pages = {1508–1527}
}

@article{Manski2000,
 ISSN = {00129682, 14680262},
 URL = {http://www.jstor.org/stable/2999533},
 author = {Charles F. Manski and John V. Pepper},
 journal = {Econometrica},
 number = {4},
 pages = {997--1010},
 publisher = {[Wiley, Econometric Society]},
 title = {Monotone Instrumental Variables: With an Application to the Returns to Schooling},
 urldate = {2026-02-18},
 volume = {68},
 year = {2000}
}

@article{Zhang2021,
  title = {Bounding Causal Effects on Continuous Outcome},
  volume = {35},
  ISSN = {2159-5399},
  url = {http://dx.doi.org/10.1609/aaai.v35i13.17449},
  DOI = {10.1609/aaai.v35i13.17449},
  number = {13},
  journal = {Proceedings of the AAAI Conference on Artificial Intelligence},
  publisher = {Association for the Advancement of Artificial Intelligence (AAAI)},
  author = {Zhang,  Junzhe and Bareinboim,  Elias},
  year = {2021},
  month = may,
  pages = {12207–12215}
}

@article{Duarte2023,
  title = {An Automated Approach to Causal Inference in Discrete Settings},
  ISSN = {1537-274X},
  url = {http://dx.doi.org/10.1080/01621459.2023.2216909},
  DOI = {10.1080/01621459.2023.2216909},
  journal = {Journal of the American Statistical Association},
  publisher = {Informa UK Limited},
  author = {Duarte,  Guilherme and Finkelstein,  Noam and Knox,  Dean and Mummolo,  Jonathan and Shpitser,  Ilya},
  year = {2023},
  month = aug,
  pages = {1–16}
}

@article{Manski1990,
 ISSN = {00028282},
 URL = {http://www.jstor.org/stable/2006592},
 author = {Charles F. Manski},
 journal = {The American Economic Review},
 number = {2},
 pages = {319--323},
 publisher = {American Economic Association},
 title = {Nonparametric Bounds on Treatment Effects},
 urldate = {2026-02-11},
 volume = {80},
 year = {1990}
}

@article{Lawlor2008,
  title = {Mendelian randomization: Using genes as instruments for making causal inferences in epidemiology},
  volume = {27},
  ISSN = {1097-0258},
  url = {http://dx.doi.org/10.1002/sim.3034},
  DOI = {10.1002/sim.3034},
  number = {8},
  journal = {Statistics in Medicine},
  publisher = {Wiley},
  author = {Lawlor,  Debbie A. and Harbord,  Roger M. and Sterne,  Jonathan A. C. and Timpson,  Nic and Davey Smith,  George},
  year = {2008},
  month = feb,
  pages = {1133–1163}
}

@article{Imbens1994,
  title = {Identification and Estimation of Local Average Treatment Effects},
  volume = {62},
  ISSN = {0012-9682},
  url = {http://dx.doi.org/10.2307/2951620},
  DOI = {10.2307/2951620},
  number = {2},
  journal = {Econometrica},
  publisher = {JSTOR},
  author = {Imbens,  Guido W. and Angrist,  Joshua D.},
  year = {1994},
  month = mar,
  pages = {467}
}

@article{Davies2018,
  title = {Reading Mendelian randomisation studies: a guide,  glossary,  and checklist for clinicians},
  ISSN = {1756-1833},
  url = {http://dx.doi.org/10.1136/bmj.k601},
  DOI = {10.1136/bmj.k601},
  journal = {BMJ},
  publisher = {BMJ},
  author = {Davies,  Neil M and Holmes,  Michael V and Davey Smith,  George},
  year = {2018},
  month = jul,
  pages = {k601}
}

@article{Angrist1996,
  title = {Identification of Causal Effects Using Instrumental Variables},
  volume = {91},
  ISSN = {1537-274X},
  url = {http://dx.doi.org/10.1080/01621459.1996.10476902},
  DOI = {10.1080/01621459.1996.10476902},
  number = {434},
  journal = {Journal of the American Statistical Association},
  publisher = {Informa UK Limited},
  author = {Angrist,  Joshua D. and Imbens,  Guido W. and Rubin,  Donald B.},
  year = {1996},
  month = jun,
  pages = {444–455}
}

@inbook{Robins2004,
  title = {Optimal Structural Nested Models for Optimal Sequential Decisions},
  ISBN = {9781441990761},
  ISSN = {0930-0325},
  url = {http://dx.doi.org/10.1007/978-1-4419-9076-1_11},
  DOI = {10.1007/978-1-4419-9076-1_11},
  booktitle = {Proceedings of the Second Seattle Symposium in Biostatistics},
  publisher = {Springer New York},
  author = {Robins,  James M.},
  year = {2004},
  pages = {189–326}
}

@article{Athey2021,
  title = {Policy Learning With Observational Data},
  volume = {89},
  ISSN = {0012-9682},
  url = {http://dx.doi.org/10.3982/ECTA15732},
  DOI = {10.3982/ecta15732},
  number = {1},
  journal = {Econometrica},
  publisher = {The Econometric Society},
  author = {Athey,  Susan and Wager,  Stefan},
  year = {2021},
  pages = {133–161}
}

@article{Kitagawa2018,
  title = {Who Should Be Treated? Empirical Welfare Maximization Methods for Treatment Choice},
  volume = {86},
  ISSN = {0012-9682},
  url = {http://dx.doi.org/10.3982/ECTA13288},
  DOI = {10.3982/ecta13288},
  number = {2},
  journal = {Econometrica},
  publisher = {The Econometric Society},
  author = {Kitagawa,  Toru and Tetenov,  Aleksey},
  year = {2018},
  pages = {591–616}
}

@article{Manski2004,
  title = {Statistical Treatment Rules for Heterogeneous Populations},
  volume = {72},
  ISSN = {1468-0262},
  url = {http://dx.doi.org/10.1111/j.1468-0262.2004.00530.x},
  DOI = {10.1111/j.1468-0262.2004.00530.x},
  number = {4},
  journal = {Econometrica},
  publisher = {The Econometric Society},
  author = {Manski,  Charles F.},
  year = {2004},
  month = jul,
  pages = {1221–1246}
}

@book{Tsiatis2019,
  title = {Dynamic Treatment Regimes: Statistical Methods for Precision Medicine},
  ISBN = {9780429192692},
  url = {http://dx.doi.org/10.1201/9780429192692},
  DOI = {10.1201/9780429192692},
  publisher = {Chapman and Hall/CRC},
  author = {Tsiatis,  Anastasios A. and Davidian,  Marie and Holloway,  Shannon T. and Laber,  Eric B.},
  year = {2019},
  month = dec 
}

@article{Murphy2003,
  title = {Optimal Dynamic Treatment Regimes},
  volume = {65},
  ISSN = {1467-9868},
  url = {http://dx.doi.org/10.1111/1467-9868.00389},
  DOI = {10.1111/1467-9868.00389},
  number = {2},
  journal = {Journal of the Royal Statistical Society Series B: Statistical Methodology},
  publisher = {Oxford University Press (OUP)},
  author = {Murphy,  S. A.},
  year = {2003},
  month = apr,
  pages = {331–355}
}

@article{Cui2020,
  title = {A Semiparametric Instrumental Variable Approach to Optimal Treatment Regimes Under Endogeneity},
  volume = {116},
  ISSN = {1537-274X},
  url = {http://dx.doi.org/10.1080/01621459.2020.1783272},
  DOI = {10.1080/01621459.2020.1783272},
  number = {533},
  journal = {Journal of the American Statistical Association},
  publisher = {Informa UK Limited},
  author = {Cui,  Yifan and Tchetgen Tchetgen,  Eric},
  year = {2020},
  month = aug,
  pages = {162–173}
}

@article{garreau2017large,
  title={Large sample analysis of the median heuristic},
  author={Garreau, Damien and Jitkrittum, Wittawat and Kanagawa, Motonobu},
  journal={arXiv preprint arXiv:1707.07269},
  year={2017}
}

@book{Kosorok2008,
  title = {Introduction to Empirical Processes and Semiparametric Inference},
  ISBN = {9780387749785},
  ISSN = {0172-7397},
  url = {http://dx.doi.org/10.1007/978-0-387-74978-5},
  DOI = {10.1007/978-0-387-74978-5},
  journal = {Springer Series in Statistics},
  publisher = {Springer New York},
  author = {Kosorok,  Michael R.},
  year = {2008}
}

@book{Steinwart2008,
  title = {Support Vector Machines},
  ISBN = {9780387772424},
  ISSN = {2197-4128},
  url = {http://dx.doi.org/10.1007/978-0-387-77242-4},
  DOI = {10.1007/978-0-387-77242-4},
  journal = {Information Science and Statistics},
  publisher = {Springer New York},
  author = {Steinwart,  Ingo and Christmann,  Andreas},
  year = {2008}
}

@article{Zhang2020,
  title = {Multicategory Outcome Weighted Margin-based Learning for Estimating Individualized Treatment Rules},
  ISSN = {1017-0405},
  url = {http://dx.doi.org/10.5705/ss.202017.0527},
  DOI = {10.5705/ss.202017.0527},
  journal = {Statistica Sinica},
  publisher = {Statistica Sinica (Institute of Statistical Science)},
  author = {Zhang,  Chong and Chen,  Jingxiang and Fu,  Haoda and He,  Xuanyao and Zhao,  Ying-Qi and Liu,  Yufeng},
  year = {2020}
}

@article{Lee2004,
  title = {Multicategory Support Vector Machines: Theory and Application to the Classification of Microarray Data and Satellite Radiance Data},
  volume = {99},
  ISSN = {1537-274X},
  url = {http://dx.doi.org/10.1198/016214504000000098},
  DOI = {10.1198/016214504000000098},
  number = {465},
  journal = {Journal of the American Statistical Association},
  publisher = {Informa UK Limited},
  author = {Lee,  Yoonkyung and Lin,  Yi and Wahba,  Grace},
  year = {2004},
  month = mar,
  pages = {67–81}
}

@book{Pearl2009,
  title = {Causality: Models,  Reasoning,  and Inference},
  ISBN = {9780521749190},
  url = {http://dx.doi.org/10.1017/CBO9780511803161},
  DOI = {10.1017/cbo9780511803161},
  publisher = {Cambridge University Press},
  author = {Pearl,  Judea},
  year = {2009},
  month = sep 
}

@article{Balke1997,
  title = {Bounds on Treatment Effects from Studies with Imperfect Compliance},
  volume = {92},
  ISSN = {1537-274X},
  url = {http://dx.doi.org/10.1080/01621459.1997.10474074},
  DOI = {10.1080/01621459.1997.10474074},
  number = {439},
  journal = {Journal of the American Statistical Association},
  publisher = {Informa UK Limited},
  author = {Balke,  Alexander and Pearl,  Judea},
  year = {1997},
  month = sep,
  pages = {1171–1176}
}

@article{Sachs2022,
  title = {A General Method for Deriving Tight Symbolic Bounds on Causal Effects},
  volume = {32},
  ISSN = {1537-2715},
  url = {http://dx.doi.org/10.1080/10618600.2022.2071905},
  DOI = {10.1080/10618600.2022.2071905},
  number = {2},
  journal = {Journal of Computational and Graphical Statistics},
  publisher = {Informa UK Limited},
  author = {Sachs,  Michael C. and Jonzon,  Gustav and Sj\"{o}lander,  Arvid and Gabriel,  Erin E.},
  year = {2022},
  month = may,
  pages = {567–576}
}

@article{Zhang2014,
  title = {Multicategory angle-based large-margin classification},
  volume = {101},
  ISSN = {1464-3510},
  url = {http://dx.doi.org/10.1093/biomet/asu017},
  DOI = {10.1093/biomet/asu017},
  number = {3},
  journal = {Biometrika},
  publisher = {Oxford University Press (OUP)},
  author = {Zhang,  C. and Liu,  Y.},
  year = {2014},
  month = jul,
  pages = {625–640}
}

@article{Cui2021,
  title = {Individualized Decision Making Under Partial Identification: ThreePerspectives,  Two Optimality Results,  and One Paradox},
  url = {http://dx.doi.org/10.1162/99608f92.d07b8d16},
  DOI = {10.1162/99608f92.d07b8d16},
  journal = {Harvard Data Science Review},
  publisher = {MIT Press - Journals},
  author = {Cui,  Yifan},
  year = {2021},
  month = jun 
}

@book{vanderVaart1996,
  title = {Weak Convergence and Empirical Processes},
  ISBN = {9781475725452},
  ISSN = {0172-7397},
  url = {http://dx.doi.org/10.1007/978-1-4757-2545-2},
  DOI = {10.1007/978-1-4757-2545-2},
  journal = {Springer Series in Statistics},
  publisher = {Springer New York},
  author = {van der Vaart,  Aad W. and Wellner,  Jon A.},
  year = {1996}
}

@book{Scholkopf2001,
author = {Scholkopf, Bernhard and Smola, Alexander J.},
title = {Learning with Kernels: Support Vector Machines, Regularization, Optimization, and Beyond},
year = {2001},
isbn = {0262194759},
publisher = {MIT Press},
address = {Cambridge, MA, USA}
}

@article{puzhang2021,
    author = {Pu, Hongming and Zhang, Bo},
    title = {Estimating Optimal Treatment Rules with an Instrumental Variable: A Partial Identification Learning Approach},
    journal = {Journal of the Royal Statistical Society Series B: Statistical Methodology},
    volume = {83},
    number = {2},
    pages = {318-345},
    year = {2021},
    month = {03},
    issn = {1369-7412},
    doi = {10.1111/rssb.12413},
    url = {https://doi.org/10.1111/rssb.12413},
    eprint = {https://academic.oup.com/jrsssb/article-pdf/83/2/318/49320528/jrsssb_83_2_318.pdf},
}
\end{refsection}
\end{document}